\documentclass[11pt]{article}
\usepackage[utf8]{inputenc}
\usepackage{algorithm}
\usepackage{algpseudocode}
\usepackage[margin=1.2in]{geometry}
\usepackage{secdot}
\usepackage{amsmath}
\usepackage{amssymb}
\usepackage{amsfonts}
\usepackage{mathrsfs,bm}
\usepackage{natbib}
\usepackage{graphicx}
\usepackage{graphics}
\usepackage{float}
\usepackage{hyperref}
\usepackage{multirow}
\usepackage{caption} 
\usepackage{enumerate}
\usepackage{bm}
\usepackage{amsthm,booktabs,color,epsfig,url}
\newtheorem{theorem}{Theorem}
\newtheorem{lemma}{Lemma}

\newtheorem{proposition}{Proposition}
\newtheorem{remark}{Remark}
\newtheorem{condition}{Condition}
\hypersetup{
	colorlinks=true,
	linkcolor=blue,
	filecolor=blue,      
	urlcolor=blue,
	citecolor=blue,
}
\algnewcommand{\Initialize}[1]{%
	\State \textbf{Initialize:}
	\Statex \hspace*{\algorithmicindent}\parbox[t]{.8\linewidth}{\raggedright #1}
}

\newcommand{\pcite}[1]{\citeauthor{#1}'s \citeyearpar{#1}}
\newcommand{\df}{\mathrm{d}}
\newcommand{\Dirac}{\mathcal{D}}
\newcommand{\cov}{\psi}
\newcommand{\tN}{N}

\usepackage[english]{babel}
\usepackage{fancyhdr}
\usepackage{lastpage}
\usepackage{verbatim}
\usepackage{pdfpages}
\usepackage{authblk}
\title{\bf{Multivariate strong invariance principle and uncertainty assessment for time in-homogeneous cyclic MCMC samplers}}
\author{Haoxiang Li}
\author{Qian Qin}

\affil{School of Statistics \\ University of Minnesota}

\date{}

\begin{document}

	\maketitle
	\thispagestyle{fancy}
\begin{abstract}
Time in-homogeneous cyclic Markov chain Monte Carlo (MCMC) samplers, including deterministic scan Gibbs samplers and Metropolis within Gibbs samplers, are extensively used for sampling from multi-dimensional distributions. 
We establish a multivariate strong invariance principle (SIP) for Markov chains associated with these samplers.
The rate of this SIP essentially aligns
with the tightest rate available for time homogeneous Markov chains.
The SIP implies the strong law of large numbers (SLLN)  and the central limit theorem (CLT), and plays an essential role in uncertainty assessments. 
Using the SIP, we give conditions under which the multivariate batch means estimator for estimating the covariance matrix in the multivariate CLT is strongly consistent. 
Additionally, we provide conditions for a  multivariate fixed volume sequential termination rule, which  is associated with the concept of effective sample size (ESS),  to be asymptotically valid.  
Our uncertainty assessment tools are demonstrated through various numerical experiments.
\end{abstract}

\section{Introduction}
An MCMC algorithm simulates an irreducible Markov chain with the target distribution serving as its stationary distribution.
Expectations with respect to the target distribution are then estimated by sample averages.
Deterministic-scan Gibbs samplers and Metropolis-within-Gibbs samplers are time in-homogeneous MCMC  algorithms characterized by transitions that cyclically update each component in a multi-dimensional state space. 
We call such samplers deterministic scan Gibbs type samplers.
This type of ``cyclic" time in-homogeneous samplers are frequently  used  to study multi-dimensional distributions, which are commonly seen in modern Bayesian studies (see e.g., the \texttt{BUGS} \citep{lunn2009bugs} software package). 

Oftentimes, one circumvents time-inhomogeneity by subsampling the underlying chain to a homogeneous subchain.
But outside some special circumstances, subsampling results in a loss in efficiency
 \citep{maceachern1994subsampling, greenwood1996outperforming, greenwood1998information}.
Indeed, time in-homogeneous cyclic samplers are more efficient than their time homogeneous counterparts in plenty of situations. 
For instance, a deterministic scan Gibbs sampler with two components converges faster than its random scan counterpart \citep{qin2022convergence}. 
Moreover, the former can outperform the latter in terms of asymptotic variance~\citep{greenwood1998information, qin2022analysis}.
Despite the usefulness of in-homogeneous samplers, there is a lack of strong approximation results and tools for uncertainty assessment associated with them.
 
We establish multivariate strong invariance principle (SIP) results for the sample averages of time in-homogeneous cyclic Markov chains. 
The SIP, also commonly referred to as strong approximation, is a convergence result that is significantly stronger than the strong law of large numbers (SLLN) and the central limit theorem (CLT).
With the help of the SIP, we prove the strong consistency of the batch means estimator for estimating the asymptotic covariance matrix in the CLT.
This in turn leads to an asymptotic valid termination rule for the Monte Carlo simulation.

Let $(\mathcal{X}, \mathcal{F})$ be a countably generated measurable space, and $\pi$ be a probability measure on $(\mathcal{X}, \mathcal{F})$. Consider a time in-homogeneous cyclic Markov chain $\{X_t\}_{t=0}^\infty$ on the state space $(\mathcal{X}, \mathcal{F})$ with stationary distribution $\pi$. The chain can start from an arbitrary initial distribution.
Let $f: \mathcal{X} \to \mathbb{R}^d$ be a measurable function, where $d$ is a positive integer, and let $\theta =  \pi(f) = \int_{\mathcal{X}}f(x)\pi(\mathrm{d}x)$ be the quantity we want to estimate.
Define the Monte Carlo sample mean estimator  as
$\hat{\theta}_{n} = \left(\sum_{t=1}^n  f(X_t)\right)/{n},$
where $n$ is the Monte Carlo sample size.
The {SLLN} holds, if with probability 1, 
$
	\hat{\theta}_{n} \to \theta \text{ as } n \to \infty.
$
The uncertainty of Monte Carlo estimator can be assessed through the CLT. 
The multivariate {CLT} holds if 
\begin{equation}\label{eqn:clt}
	\sqrt n (\hat{\theta}_{n} - \theta) \xrightarrow{d} \mathcal{N}(0, \Sigma), \text{ as } n \to \infty,
\end{equation}
for some $d \times d$ positive definite matrix $\Sigma$, and $\Sigma$ is called asymptotic covariance matrix.
A strong invariance principle (SIP) holds if on a suitable probability space, one can redefine the sequence  $\{X_t\}_{t=0}^\infty$ (while preserving its law) 
together with $\{C(t)\}_{t=1}^\infty$, where 
$\{C(t)\}_{t=1}^\infty$ are i.i.d. $d$ dimensional standard normal random vectors, such that for almost all sample paths $\omega$, for $n$ large enough, 
\begin{equation}\label{eqn:sip0}
	\left\|\sum_{t=1}^n(f(X_t(\omega)) - \theta) - \Gamma  B(n)(\omega) \right\| \leq M(\omega) \phi (n),
\end{equation}
where $B(n)=\sum_{t=1}^n C(t)$, $M$ is a finite random variable, $\Gamma$ is a $d\times d$ constant matrix satisfying $\Gamma \Gamma^\top = \Sigma $,  $\|\cdot\|$ is the Euclidean norm, and the rate $\phi (n)$ is a non-negative increasing function on $\mathbb{N}_+:= \{1,2,\dots\}$ with $\phi (n)/\sqrt n \to 0$ as $n \to \infty$. 
The SIP ensures the SLLN and the CLT. The SIP also plays a crucial role in uncertainty assessments. 

There are a number of existing results that establish the SIP for time homogeneous Markov chains.
For stationary $\alpha$-mixing processes, 
including polynomially ergodic stationary Markov chains, 
 \cite{kuelbs1980almost} derived \eqref{eqn:sip0} with  $\phi(n) = n^{1/2-\lambda}$ for some unknown $\lambda>0$.
\cite{merlevede2012strong} constructed univariate, i.e., $d=1$,  version of \eqref{eqn:sip0} with $\phi(n) = n^{1/\lambda}(\log n)^{1/2-1/\lambda}$, $\lambda\in[2,3]$, for stationary $\alpha$-mixing processes, under conditions related to the quantile function and the rate of mixing.
Under a 1-step minorization condition and when $d=1$, \cite{jones2006fixed} obtained explicit SIP rates for geometrically ergodic Markov chains, and  the result by  \cite{csaki1995additive} forms the basis of their result. 
For a rich class of time  in-homogeneous cyclic Markov chains, when the function $f$ has a finite $2+\gamma+\gamma^*$ moment under the distribution $\pi$, and the chain is polynomially ergodic of order larger than $(1+\gamma+(2+\gamma)^2/\gamma^*)$ for some $\gamma>0$ and $\gamma^*>0$, we establish the SIP \eqref{eqn:sip0} with $\phi(n) = n^{\rho}$ for $\rho > \max\{1/(2+\gamma), 1/4\}$. This rate essentially aligns with  the tightest result achieved for homogeneous Markov chains \citep{csaki1995additive}.


We prove the SIP using a new regenerative construction suitable for cyclic time in-homogeneous, especially Gibbs type, chains.
This construction allows us to apply an existing SIP result for a certain class of stationary dependent process \citep{liu2009strong}.
Some of the techniques herein are adopted from the recent work of \cite{banerjee2022multivariate} who studied homogeneous chains.


We provide tools for assessing the uncertainty of Monte Carlo estimators formed from time in-homogeneous cyclic MCMC samplers.
The uncertainty of Monte Carlo estimation can be quantified by combining the CLT and a consistent estimator of the asymptotic covariance matrix in the CLT.
The batch means estimator is frequently employed for estimating the asymptotic covariance matrix.
The consistency of the batch means estimator can be established through the SIP.
This is demonstrated by \cite{jones2006fixed}, \cite{flegal2010batch}, and \cite{vats2019multivariate}, who studied time homogeneous chains.
For  time in-homogeneous cyclic Markov chains, we establish the strong consistency of the multivariate batch means estimator through the SIP.

Determining the necessary sampling length for achieving a certain level of accuracy is a fundamental question in MCMC.
Given a stochastic simulation, \cite{glynn1992asymptotic} studied a class of  sequential termination rules that can be used to obtain confidence sets with a fixed volume.
\cite{jones2006fixed} applied these sequential termination rules to  time homogeneous Markov chains in a univariate setting. 
\cite{vats2019multivariate} applied the rules to a multivariate setting for time homogeneous Markov chains and 
connected the rules to the concept of effective sample size (ESS). 
Utilizing the SIP, we study \pcite{vats2019multivariate} version of the rule in the context of time in-homogeneous cyclic Markov chains, and give conditions for the rule to be asymptotically valid.

We apply our results for uncertainty assessment and termination rule to two numerical examples.
In each example, we show that a time in-homogeneous sampler is more efficient than its natural homogeneous counterpart. 

The rest of the paper is arranged as follows. In Section \ref{sec:pre}, we introduce the preliminary background for our work. 
We present our main results in Section \ref{sec:main}.
These include an SLLN, a CLT, and an SIP for time in-homogeneous cyclic Markov chains, and tools for uncertainty assessment based on them.
Numerical experiments are provided in Section \ref{sec:num}. 

\section{Preliminaries}\label{sec:pre}
\subsection{Time homogeneous Markov chains} \label{ssec:pret}

A Markov chain with state space $(\mathcal{X}, \mathcal{F})$ is {time homogeneous} if its transition law stays the same for all $t\in\mathbb{N}_+$.
Let  $K: \mathcal{X} \times \mathcal{F} \to [0,1]$ be the  transition kernel of the chain. For a signed measure $\mu$ and a measurable function $f$ on $(\mathcal{X}, \mathcal{F})$, with a slight abuse of notation, let
$\mu (f) = \int_{\mathcal{X}} f(x)\mu(\mathrm{d}x)$.
The kernel $K$ can act on $\mu$ and $f$ as follows:
$\mu K(A) = \int _{\mathcal{X}} K(x, A) \mu(\mathrm{d}x),$ $A \in \mathcal{F}$,  $Kf(x) = \int_{\mathcal{X}} f(y) K(x, \mathrm{d}y),$ $x \in \mathcal{X},$ provided that the integrals are well defined.
Given two measures $\mu$ and $\nu$ on $(\mathcal{X}, \mathcal{F})$, define $\|\mu(\cdot) - \nu(\cdot)\|_{\text{TV}}$ as the total variation distance between $\mu$ and $\nu$.
Assume that the chain has a {stationary distribution} $\pi$, i.e., $\pi K(\cdot) = \pi(\cdot).$

Two transition kernels $K$ and $K'$ on $(\mathcal{X}, \mathcal{F})$ can be multiplied through the formula $K K'(x, \cdot) = \int_{\mathcal{X}} K(x, \df y) K'(y, \cdot)$.
In particular, for $t \in \mathbb{N}_+$, $K^t(x, \cdot)$ gives the conditional distribution of the $t$th element of a chain associated with $K$ assuming that the chain starts from~$x$.
The transition kernel $K$ is {Harris ergodic} if
$\lim_{t\to \infty }\| K^t(x,\cdot) - \pi(\cdot)\|_{\text{TV}} = 0$ for all $x\in\mathcal{X}$. 
Harris ergodicity is equivalent to the chain being $\phi$-irreducible, aperiodic, and Harris recurrent; it is also equivalent to $\lim_{t\to \infty } \|\nu K^t(\cdot) - \pi(\cdot) \|_{\text{TV}} = 0$ for an arbitrary initial distribution $\nu$ \citep[][Section 6.3]{nummelin2004general}.
Let $\rho:\mathcal{X} \to [0,\infty]$ be a function satisfying $\pi(\rho) < \infty$, and let $\eta(t)$ be a non-negative decreasing function on $\mathbb{N}_+$ such that for all $x\in\mathcal{X}$ and $t \in \mathbb{N}_+$,
\begin{equation}\label{eqn:ergodic}
	\| K^t(x, \cdot) - \pi(\cdot)\|_{\text{TV}} \leq \rho(x) \eta(t).
\end{equation}
The transition kernel $K$ is {polynomially ergodic of order $q$} if it is Harris ergodic and \eqref{eqn:ergodic} holds with $\eta(t) = t ^{-q}$, where $q$ is a positive real number. If the kernel is Harris ergodic, and  \eqref{eqn:ergodic} holds with $\eta(t) =   c^{t}$ for some $  c<1$, then the kernel is {geometrically ergodic}. The kernel is {at least polynomially ergodic of order $q$} if it is either polynomially ergodic of order at least $q$ or geometrically ergodic.
Polynomial ergodicity is usually established through a set of drift and minorization conditions \citep[see e.g.,][Theorem 1 and Section 1.2]{fort2003polynomial}, and geometric ergodicity can be established by a slightly stronger version of drift and minorization conditions \citep[see e.g.,][]{roberts2004general}.

\subsection{Time in-homogeneous cyclic Markov chains}\label{ssec:gibbs}
Let $\pi$ be a probability measure on space $(\mathcal{X}, \mathcal{F})$. Consider $k$ transition kernels $\{K_i\}_{i=1}^k$, where at least two of them are distinct. 
Assume that  leave the distribution $\pi$ invariant,  i.e., $\pi K_i(\cdot) = \pi(\cdot)$, for all $i\in\{1,\dots ,k\}$. A {time in-homogeneous $k$-cyclic MCMC sampler} applies the $k$ transition kernels cyclically, i.e., its transition kernel at time ${kj+i-1}$ is  $K_i,$ $i \in \{ 1,\dots ,k\}$, $j \in \mathbb{N} := \{0,1,2,\dots\}$.
Let $\{X_t\}_{t=0}^\infty$ be a Markov chain associated with the time in-homogeneous $k$-cyclic MCMC sampler, so that $P(X_{kj+i} \in A \mid X_{kj+i-1} = x) = K_i(x, A)$ for $x \in \mathcal{X}$ and $A \in \mathcal{F}$.  
The chain is not time homogeneous, but  we have $k$ homogeneous subchains by subsampling.  For $i\in \{ 1,\dots ,k\}$, define $\{Z_t^i\}_{t=0}^{\infty} = \{X_{kt+i-1}\}_{t=0}^{\infty}$ as the {$i$th homogeneous subchain} of $\{X_t\}_{t=0}^\infty$. 
The first homogeneous subchain has transition kernel $\tilde{K}_{1} = K_1K_2\cdots K_k$.
The {block chain} $\{Z_{t}^{\scriptsize\mbox{blo}}\}_{t=0}^{\infty} = \{(X_{kt}, \dots, X_{kt+k-1})\}_{t=0}^{\infty}$ also forms a time homogeneous Markov chain. 

The most common time in-homogeneous $k$-cyclic MCMC samplers are deterministic scan Gibbs and Metropolis within Gibbs samplers. 
For example, when $k=2$ and $X=(U,V)$ is a random vector distributed as $\pi$, the deterministic scan Gibbs chain targeting $\pi$
has the form $\{(U_0, V_0), (U_1, V_0), (U_1, V_1), (U_2, V_1), \dots \}$. The first homogeneous subchain is of the form $\{(U_0, V_0), (U_1, V_1), \dots\}$, and the block chain is of the form $$\left\{
((U_0, V_0), (U_1, V_0)), ((U_1, V_1), (U_2, V_1)), \dots
\right\}.$$
Note that in some existing works, the name ``deterministic scan Gibbs sampler" is reserved for the first homogeneous subchain $\{Z_t^1\}_{t=0}^{\infty}$.

The lemma below gives useful relationships regarding the convergence properties of the block chain and subchains.

\begin{lemma}\label{lem:1}
	Suppose that the time in-homogeneous chain $\{X_t\}_{t=0}^\infty$ possesses a time homogeneous subchain whose transition kernel is Harris ergodic.
	Then the transition kernels of the $j$th homogeneous subchain, $j\in\{1,\dots,k\}$,  and the block chain are Harris ergodic. 
	Suppose that the time in-homogeneous chain $\{X_t\}_{t=0}^\infty$ possesses a time homogeneous subchain whose transition kernel is polynomially ergodic of order $q$ (resp. geometrically ergodic).
	Then the transition kernels of the $j$th homogeneous subchain, $j\in\{1,\dots,k\}$,  and the block chain are polynomially ergodic of order $q$ (resp. geometrically ergodic). 
\end{lemma}

Lemma \ref{lem:1} holds following the de-initializing process argument in~\pcite{roberts2001markov} Theorem 1.

\section{Main Results}\label{sec:main}

\subsection{The SLLN and the multivariate CLT}\label{ssec:sllnclt} 
We begin with some elementary convergence results for sample averages of a time in-homogeneous cyclic chain.

Recall that $\{X_t\}_{t=0}^\infty$ is a time in-homogeneous cyclic Markov chain, started from an arbitrary initial distribution, and $\hat{\theta}_n = (\sum_{t=1}^n f(X_t))/n$ is the sample mean estimator estimating $\theta = \pi (f)$.

\begin{theorem}\label{thm:slln}
	Suppose that the time in-homogeneous chain $\{X_t\}_{t=0}^\infty$ possesses a time homogeneous subchain whose transition kernel
	is Harris ergodic, and $\pi( \|f\|_1 ) <\infty$, where $\|f(x)\|_1 = |f_1(x)|+\dots+|f_d(x)|$ for $x \in \mathcal{X}$.  Then  with probability 1, 
	$
	\hat{\theta}_{n} \to \theta \text{ as } n \to \infty.
	$
\end{theorem}

To prove Theorem \ref{thm:slln},  we apply \pcite{meyn2005markov} Theorem 17.0.1 on the homogeneous subchains.
Further details can be found in Appendix \ref{proof:thmslln}.

For $j\in\{0,\dots ,k-1\}$ and $l\in\mathbb{Z}$, define the $(j,l)$th autocovariance matrix as
\begin{equation}\label{eqn:xi}
	\begin{split}
		& \cov(j,l) = E_{\pi} \{(f(X_j) - \theta)(f(X_{j+l}) -\theta)^\top\}, \ l\geq 0, \text{ and }\\
		& \cov(j, l) = E_{\pi} \{(f(X_{j-l}) -\theta) (f(X_j) - \theta)^\top\} , \ l<0,
	\end{split}
\end{equation}
where  $E_{\pi}(\cdot)$ denotes the expectation of a function of $\{X_t\}_{t=0}^{\infty}$  if $X_0 \sim \pi$.
To be precise, for an initial distribution $\nu$ and a random vector $V(X_0, X_1, \dots)$, we use $E_{\nu} V(X_0, X_1, \dots)$ to denote the expectation of $V(X_0', X_1', \dots)$ where $\{X_t'\}_{t=0}^{\infty}$ is a Markov chain that has the same transition law as $\{X_t\}_{t=0}^{\infty}$, and $X_0' \sim \nu$.
Clearly, $\cov(j,l) = \cov(j, -l)^\top$.
Different from the time homogeneous situation, when $j_1\neq j_2$, $\cov(j_1,l)\neq \cov(j_2,l)$.
The theorem below provides the multivariate CLT, and gives an exact form of  the  covariance matrix $\Sigma$ in the CLT.

\begin{theorem}\label{thm:clt}
	Suppose that the time in-homogeneous chain $\{X_t\}_{t=0}^\infty$ possesses a time homogeneous subchain whose transition kernel
	is at least polynomially ergodic of order $q\geq (1+\gamma)(1+2/\gamma^*)$ and $\pi( \|f\|^{2+\gamma^*} )<\infty$, for some $\gamma>0$, and $\gamma^*>0$.  
	Then for $a$ and $b$ in $\{1,\dots,d\}$, $\sum_{l=-\infty}^{+\infty}\sum_{j=0}^{k-1} |\cov_{a,b}(j, l)|<\infty$, where $\cov_{a,b}(j, l)$ is the $(a,b)$th element of $\cov(j, l)$.
	Define 
		\begin{equation}\label{eqn:varsigma}
		\Sigma = \sum_{l=-\infty}^{+\infty} \sum_{j=0}^{k-1}\frac{1}{k}\cov(j,l).
	\end{equation}
	Assume further that $\Sigma$ is positive definite. Then
	the multivariate CLT \eqref{eqn:clt} holds with asymptotic covariance matrix $\Sigma$.
\end{theorem}

The proof of Theorem \ref{thm:clt} relies on an application of \pcite{jones2004markov} Theorem 9 (2) on the block chain, and \pcite{kuelbs1980almost} Lemma 2.1. Further details are given in Appendix \ref{proof:clt}. 

\subsection{The multivariate strong invariance principle}\label{ssec:sip} 

\subsubsection{Sufficient conditions for an SIP} \label{sssec:sip}

To establish the SIP using a regeneration argument, we shall impose the following condition on $\{X_t\}_{t=0}^\infty$.  


\begin{condition}\label{con:min}
	There exist a measurable space $(\mathcal{X}_U, \mathcal{F}_U)$, a Markov transition kernel $K_U: \mathcal{X}_U \times \mathcal{F}_U \to [0,1]$, an integer $k_0 \in \mathbb{N}_+$, a measurable function $g_0: \mathcal{X}^{k_0} \to \mathcal{X}_U$, and measurable functions $g_i: \mathcal{X}_U^2 \to \mathcal{X}$, $i = 1,\dots,k$, that satisfy the following conditions:
	\begin{enumerate}
		\item [(i)] 
		Regardless of the initial distribution of $\{X_t\}_{t=0}^{\infty}$, the following hold: the process $$U_t = g_0(X_{kt}, X_{kt+1}, \dots, X_{kt+k_0-1}), \; t \in \mathbb{N},$$ forms a homogeneous Markov chain following the transition law $K_U$;
		moreover,  for $i\in\{1,\dots,k\}$ and $t\in\mathbb{N}$, $X_{tk + i} = g_i(U_t, U_{t+1})$.
		\item [(ii)] (1-step minorization) The transition kernel $K_U$ has a stationary distribution $\pi_U$.
		Moreover, there exists a function $h:\mathcal{X}_U\to[0,1]$ with $\pi_U(h)>0$ and a probability measure $\mu(\cdot)$ on $(\mathcal{X}_U,\mathbb{F}_U)$ such that
		\begin{equation}\label{eqn:minorization}
			K_U(u, A) \geq h(u)\mu(A), \, \quad u \in \mathcal{X}_U, \; A\in\mathcal{F}_U.
		\end{equation}	
	\end{enumerate}
\end{condition}


Condition \ref{con:min}, while seemingly technical, is satisfied for a variety of samplers, including the most common deterministic scan Gibbs-type algorithms. 
Examples will be given after we state Theorem~\ref{thm:sip}, the main result of this section below.

\begin{theorem}\label{thm:sip}
	Assume that Condition \ref{con:min} holds. 
	Suppose that $K_U$ is at least polynomially ergodic of order 
	$q>1+\gamma + (2+\gamma)^2/\gamma^*$ 
	and $ \pi (\|f\|^{2+\gamma + \gamma^*} ) <\infty$ for some $\gamma>0$ and $\gamma^*>0$. 
	Finally, assume that $\Sigma$ in \eqref{eqn:varsigma} is positive definite.
	Then, for any $ \rho > \max\{1/(2+\gamma), 1/4\}$,  on a richer probability space,
	an SIP \eqref{eqn:sip0} holds for $\{X_t\}_{t=0}^{\infty}$ with rate $\phi(n) = n^ \rho$.
\end{theorem}

\begin{remark}
	We shall establish Theorem~\ref{thm:sip} when $\{X_t\}_{t=0}^{\infty}$ has an arbitrary initial distribution.
\end{remark}

Establishing Theorem \ref{thm:sip} would require some work, and we will present the proof in Sections~\ref{sssec:regenerative} and~\ref{sssec:sip-proof}.
We now list some common time in-homogeneous cyclic MCMC samplers that satisfy Condition \ref{con:min}.

When the Markov chain $\{X_t\}_{t=0}^\infty$ is time homogeneous, one can take $\{U_t\}_{t=0}^\infty$ = $\{X_t\}_{t=0}^\infty$. 
Then Condition \ref{con:min} is satisfied if the transition kernel of $\{X_t\}_{t=0}^\infty$ satisfies Condition \ref{con:min}-(ii).
This is a common assumption in works establishing the SIP for time homogeneous Markov chains \citep{csaki1995additive,jones2006fixed}.

Next, we consider deterministic scan Gibbs and Metropolis within Gibbs samplers.
 For $i \in \{1, \dots, k\}$, let $(\mathcal{X}_i, \sigma(\mathcal{X}_i))$ be a measurable space.
 Let $\mathcal{X} = \mathcal{X}_1\times \dots\times \mathcal{X}_k$ and $\mathcal{F} = \sigma(\mathcal{X}_1)\times\dots\times\sigma(\mathcal{X}_k)$. 
 For any $x \in \mathcal{X}$ and $i\in\{1,\dots, k\}$, denote by $x^{\{i\}}$ the $i$th component of~$x$, so that $x = (x^{\{1\}}, \dots, x^{\{k\}})$.
Suppose that $K_i$ is of the form
 \[
 K_i((x^{\{1\}}, \dots, x^{\{k\}}), \df (y^{\{1\}}, \dots, y^{\{k\}})) = H_i(x, \df y^{\{i\}}) \, \prod_{j \neq i} \Dirac_{x^{\{j\}}}(\df y^{\{j\}}),
 \]
 where $\Dirac_{x^{\{j\}}}: \sigma(\mathcal{X}_j) \to [0,1]$ is the point mass (Dirac measure) at $x^{\{j\}} \in \mathcal{X}_j$, and $H_i: \mathcal{X} \times \sigma(\mathcal{X}_i) \to [0,1]$ is a Markov transition kernel.
In this case, $X_{kt + i}$ and $X_{kt+i-1}$ differ only in the $i$th coordinate for $i \in \{1,\dots,k\}$ and $t \in \mathbb{N}$.
Take $\{U_t\}_{t=0}^\infty = \{X_{kt}\}_{t=0}^\infty$,  i.e., $\{U_t\}_{t=0}^{\infty}$ is the first homogeneous subchain of $\{X_t\}_{t=0}^{\infty}$. 
For $x \in \mathcal{X}$ and $i \in \{1,\dots,k\}$, let $x^{i-} = (x^{\{1\}}, \dots, x^{\{i-1\}}, x^{\{i\}})$  and $x^{i+} = (x^{\{i\}}, x^{\{i+1\}}, \dots, x^{\{k\}})$.
Then, for $t\in\mathbb{N}$ and $i\in\{1,\dots,k\}$, $X_{tk+i}^{i-} = U_{t+1}^{i-}$, and, if $i \leq k-1$, $X_{tk+i}^{(i+1)+} = U_{t}^{(i+1)+}$.
Then Condition \ref{con:min}-(i) holds with $K_U= K_1K_2\cdots K_k$.
Suppose, for example, that the $\mathcal{X}_i$'s are Euclidean spaces, and that
\[
H_i(x, A) \geq \int_A h_i(x, y^{\{i\}}) \, \df y^{\{i\}}, \quad x \in \mathcal{X}, \; A \in \sigma(\mathcal{X}_i),
\]
where $h_i: \mathcal{X} \times \mathcal{X}_i \to (0,\infty)$ is a strictly positive lower semicontinuous function.
Then $K_U$ satisfies Condition \ref{con:min}-(ii).
Evidently, this is satisfied by a wide range of common Gibbs-type samplers.

Consider another example of a time in-homogeneous cyclic sampler.
Suppose now that $k=2$ and $K'_1$ and $K'_2$ are Gibbs updates. 
To be specific, assume that $\pi$ is the joint distribution of some random element $X = (X^{\{1\}}, X^{\{2\}})$, where $X^{\{i\}}$ is $\mathcal{X}_i$-valued for $i = 1,2$.
Let $\pi_2(x^{\{1\}}, \cdot)$ be the conditional distribution of $X^{\{2\}}$ given $X^{\{1\}} = x^{\{1\}} \in \mathcal{X}_1$, and let $\pi_1(x^{\{2\}}, \cdot)$ be the conditional distribution of $X^{\{1\}}$ given $X^{\{2\}} = x^{\{2\}} \in \mathcal{X}_2$.
Let
\[
\begin{aligned}
K'_1((x^{\{1\}}, x^{\{2\}}), \df (y^{\{1\}}, y^{\{2\}})) = \pi_1(x^{\{2\}}, \df y^{\{1\}}) \, \Dirac_{x^{\{2\}}}(\df y^{\{2\}}), \\ 
K'_2((x^{\{1\}}, x^{\{2\}}), \df (y^{\{1\}}, y^{\{2\}})) = \pi_2(x^{\{1\}}, \df y^{\{2\}}) \, \Dirac_{x^{\{1\}}}(\df y^{\{1\}}).
\end{aligned}
\]
One may apply the update $K'_1$ multiple, say $k_1$, times before moving on to $K'_2$.
This can be beneficial if the computation cost of applying $K'_1$ is much less than that of applying $K'_2$~\citep{qin2022analysis}. 
We call the resulting sampler a modified deterministic scan Gibbs sampler.
We mandate that each update constitutes an iteration of the sampler, so this is a time in-homogeneous $k$-cyclic sampler, where $k=k_1+1$. 
We can specify the $k$ transition kernels of this sampler as follows: for $i \in \{1,\dots,k_1\}$, $K_i = K_1'$, and $K_k = K_2'$.
Denote by $\{X_t\}_{t=0}^{\infty}$ the underlying Markov chain.
For $t\in\mathbb{N}$, let 
$U_t = ( X_{tk+1}^{\{1\}}, \dots, X_{tk+k_1}^{\{1\}}, X_{tk}^{\{2\}})$.
Then $\{U_t\}_{t=0}^{\infty}$ is a homogeneous Markov chain with state space $\mathcal{X}_1^{k_1} \times \mathcal{X}_2$ whose transition kernel is
\begin{equation} \label{eqn:KU}
K_U((u^{\{1\}}, \dots, u^{\{k_1\}}, u^{\{k\}}), \df (v^{\{1\}}, \dots, v^{\{k_1\}}, v^{\{k\}})) = \pi_2(u^{\{k_1\}}, \df v^{\{k\}}) \prod_{i=1}^{k_1} \pi_1(v^{\{k\}}, \df v^{\{i\}}).
\end{equation}
Moreover, for $t\in\mathbb{N}$, $X_{tk+i} = (U_t^{\{i\}}, U_t^{\{k\}})$ if $i \in \{1,\dots,k_1\}$, and $X_{tk + k} = (U_t^{\{k_1\}}, U_{t+1}^{\{k\}})$.
Hence, Condition~\ref{con:min}-(i) is satisfied.
Condition \ref{con:min}-(ii) holds if there exists a function $h_*: \mathcal{X}_1 \to [0,\infty]$ and a probability measure $\nu(\cdot)$ on $\mathcal{X}_2$ such that $$\int_{\mathcal{X}} h_*(x^{\{1\}}) \, \pi(\df (x^{\{1\}}, x^{\{2\}})) > 0,$$ and that
\[
\pi_2(x^{\{1\}}, A) \geq h_*(x^{\{1\}}) \nu(A) , \quad x^{\{1\}} \in \mathcal{X}_1 , \; A \in \sigma(\mathcal{X}_2).
\]

The convergence properties of the chain $\{U_t\}_{t=0}^{\infty}$ is closely related to those of $\{X_t\}_{t=0}^{\infty}$. 
Indeed, we have the following lemma, proved in Appendix \ref{proof:ergodic}.

\begin{lemma} \label{lem:ergodic}
	Assume that Condition \ref{con:min}-(i) holds.
	If $K_U$ is Harris ergodic, then $\tilde{K}_1$, the transition kernel of
	the homogeneous subchain $\{Z_t^1\}_{t=0}^{\infty}$,
	is Harris ergodic.
	If $K_U$ is polynomially ergodic of order $q$ (resp. geometrically ergodic), then $\tilde{K}_1$ is polynomially ergodic of order~$q$ (resp. geometrically ergodic).
\end{lemma}

\begin{remark} \label{rem:ergodic}
	Conversely, in the three examples above, $\tilde{K}_1$ being Harris ergodic implies that $K_U$ is Harris ergodic;
	$\tilde{K}_1$ being polynomially ergodic of order~$q$ (resp. geometrically ergodic) implies that $K_U$ is polynomially ergodic of order~$q$ (resp. geometrically ergodic).
	Indeed, when $\{X_t\}_{t=0}^{\infty}$ is homogeneous or a Metropolis within Gibbs chain, $\tilde{K}_1 = K_U$.
	When $\{X_t\}_{t=0}^{\infty}$ corresponds to the modified deterministic scan Gibbs sampler, the assertion can be verified through a de-initialization argument.
\end{remark}

\subsubsection{A regenerative construction} \label{sssec:regenerative}

The proof of Theorem \ref{thm:sip} relies on a regenerative construction, which we describe below.

Assume that Condition~\ref{con:min} holds for a Markov transition kernel $K_U: \mathcal{X}_U \times \mathcal{F}_U \to [0,1]$.
Define the residual measure as
$$
R(u,\mathrm{d}v) = 
\begin{cases}
	\frac{1}{1-h(u)} [K_U(u, \mathrm{d}v) - h(u)\mu(\mathrm{d}v)] , & h(u) < 1\\
	\mu(\mathrm{d}v), & h(u) = 1.
\end{cases}
$$
Then 
$$
K_U(u, \mathrm{d}v) = h(u)\mu(\mathrm{d}v) + \{1-h(u)\}R(u, \mathrm{d}v).
$$

Enriching the underlying probability space if necessary, one may define a sequence of Bernoulli random variables $\{\delta_t\}_{t=0}^{\infty}$ conditioning on $\{U_t\}_{t=0}^{\infty}$ through the following mechanism:
Given realization of $\{U_t\}_{t=0}^{\infty}$, sequentially generate $\delta_t$ for $t \in \mathbb{N}$ according to the Bernoulli distribution whose success probability is the Radon-Nikodym derivative of the measure $h(U_t) \mu(\cdot)$ with respect to $K_U(U_t, \cdot)$ evaluated at $U_{t+1}$.
It can then be shown that 
\[
\{\tilde{U}_t\}_{t=0}^\infty = \{(U_{0}, \delta_0), (U_{1}, \delta_1),  (U_{2}, \delta_2), \dots\}
\]
is a time homogeneous Markov chain.
This is called a split chain.
Given $\{(U_t, \delta_t)\}_{t=0}^{n-1}$ and $U_n$, the indicator $\delta_n$ is 1 with probability $h(U_n)$ and 0 otherwise.
Given $\{(U_t, \delta_t)\}_{t=0}^n$, $U_{n+1} \sim \mu(\cdot)$ if $\delta_n = 1$ and $U_{n+1} \sim R(U_n, \cdot)$ if $\delta_n = 0$.
See Section 3 of \cite{mykland1995regeneration}.

For $t \in \mathbb{N}$, given $\{U_i, \delta_i\}_{i < t}$, if $\delta_t=1$, then $U_{t+1}\sim\mu(\cdot)$  and does not depend on $U_{t}$.
The set of $t$ for which $\delta_{t-1}=1$ is called regeneration times, which can be defined  by $0=T_0 < T_1<\dots$ with $T_{i+1}=\inf\{t>T_i: \delta_{t-1}=1\}$.
Let $\tau_{i} = T_{i+1}-T_{i}$, $i\in\mathbb{N}$, be the length from the $i$th regeneration to the $(i+1)$th regeneration. 
The $\tau_i$s, $i\geq 1$,  are i.i.d. random variables \citep[see, e.g.,][]{mykland1995regeneration}. 
The following two lemmas highlight some important properties of the regeneration times.

\begin{lemma}[\cite{mykland1995regeneration} Theorem 1]\label{lem:finite}
Assume that $K_U$ is Harris ergodic.
Then regardless of the distribution of $(U_0, \delta_0)$, $T_1$ is finite almost surely.
\end{lemma}

The distribution of $\tau_1$  does not depend on the initial distribution. The lemma below provides conditions for finite moments.
\begin{lemma}\label{lem:tau}
	Suppose that $K_U$ is 	at least polynomially ergodic of order $q$ for some $q>1$. Then $E\tau_1^{q'}<\infty$ for all $q'\in[0,q+1)$.
\end{lemma}
Lemma \ref{lem:tau} is proved using \pcite{bolthausen1982berry} Lemma 3. 
The details are provided in Appendix \ref{proof:tau}.

For the remainder of this subsection, assume that $K_U$ is Harris ergodic, which implies that the transition kernel of the split chain is also Harris ergodic.
For  $i\in\mathbb{N}_+$, let 
\begin{equation}\label{eqn:Delta}
\Delta_i = (U_{T_{i}},...,U_{T_{i+1}-1}, \tau_i)
\end{equation}
By the regenerative structure of the split chain, the $\Delta_i$s are i.i.d. random elements and the distribution of $\Delta_1$ does not depend on the initial distribution of the split chain~\citep[see, e.g.,][Section 5.3]{nummelin2004general}.

Recall that $f: \mathcal{X} \to \mathbb{R}^d$ is a measurable function, and $\theta = \pi(f)$.
For $i\in\mathbb{N}_+$, let 
$$Y_i = \sum_{t=kT_{i}+1}^{kT_{i+1}} f(X_t) = \sum_{t=T_i}^{T_{i+1}-1} \sum_{j=1}^k f(g_j(U_t, U_{t+1})).$$
Then $Y_i$ is a function of $\Delta_i$ and $\Delta_{i+1}$.
Define $\Xi_Y = E Y_1$ and $\Xi_{\tau} =  E \tau_1$. 
By Kac's theorem~\citep[Theorem 10.2.2]{meyn2005markov}, $\Xi_{\tau} = 1/(\pi_U(h))<\infty$.  Using an argument similar to the discussions in \pcite{hobert2002applicability} Section 2, we have Lemma \ref{lem:EY} below. The details of the proof is provided in Appendix \ref{proof:EY}.

\begin{lemma}\label{lem:EY}
	If $\pi(\|f\|_1)<\infty$, then $\Xi_Y = k\Xi_{\tau}\theta<\infty$.
\end{lemma}

Assume that $\pi(\|f\|_1)<\infty$.
For $i\in\mathbb{N}_+$,
define $$\tilde{Y}_i = Y_i - {\tau_{i}\Xi_Y}/{\Xi_{\tau}} = Y_i - k \theta \tau_i.$$
Then there exists an $\mathbb{R}^d$-valued measurable function $\tilde{g}$ such that 
 \begin{equation}\label{eqn:Y}
\tilde{Y}_i = \tilde{g}(\Delta_{i}, \Delta_{i+1}) := 
\sum_{t=T_i}^{T_{i+1}-1} \left\{\sum_{j=1}^k f(g_j(U_t, U_{t+1}))\right\} - k\theta \tau_i.
\end{equation}
Clearly, sequence $\{\tilde{Y}_i\}_{i=1}^{\infty}$  is a mean-zero stationary 1-dependent sequence.
The following two lemmas list some
important properties of $\{\tilde{Y}_t\}_{t=1}^{\infty}$.

\begin{lemma}\label{lem:4}
	Suppose that $K_U$ is
	at least polynomially ergodic of order $q> 1+\gamma + (2+\gamma)^2/\gamma^*$, and $\pi( \|f\|^{2+\gamma + \gamma^*}) <\infty$ for some $\gamma\geq 0$ and $\gamma^*>0$. Then 
	\[
	E \|Y_1\|^{2+\gamma} \leq E \left( \sum_{t = kT_1 + 1}^{kT_2} \|f(X_t)\| \right)^{2+\gamma} < \infty, \quad E\left\|\tilde{Y}_1\right\|^{2+\gamma}<\infty.
	\]
\end{lemma}

The proof of Lemma \ref{lem:4} is given in Appendix \ref{proof:lem4}. We use some  arguments similar to the proofs of \pcite{banerjee2022multivariate} Lemma 3 and Lemma 6.

\begin{lemma}\label{lem:vfvar}
	Assume that the conditions in Theorem \ref{thm:sip} hold. Then $$\Sigma_Y :=\lim_{n\to\infty} n^{-1} \mathrm{Var} \left(\sum_{t=1}^n \tilde{Y}_t \right),$$ where $\mbox{Var}(\cdot)$ returns the variance of a random vector, is positive definite.
\end{lemma}

We prove Lemma \ref{lem:vfvar} by contradiction. See Appendix \ref{proof:vfvar} for details. Part of the proof uses an argument from Section 17.2.2 of \cite{meyn2005markov}.
The proof relies on Lemma~\ref{lem:error} in Section~\ref{sssec:sip-proof}, which is of course developed independently of Lemma~\ref{lem:vfvar}.

\subsubsection{Establishing the SIP} \label{sssec:sip-proof}

To prove the SIP in Theorem \ref{thm:sip}, we make use of \pcite{liu2009strong} Theorem 2.1.
The original result is somewhat complicated, so we will state only a simplified and less general version of it.

\begin{lemma} \citep[][Theorem 2.1, simple version]{liu2009strong}\label{lem:liu}
	Let $2 < \varrho < 4$.
	Let $\varepsilon_1, \varepsilon_2, \dots$ be i.i.d. random elements and suppose that, for $n \in \mathbb{N}_+$, $\tilde{W}_n = g(\varepsilon_n, \varepsilon_{n+1})$, where $g$ is some $\mathbb{R}^d$-valued measurable function.
	Suppose further that $E \tilde{W}_1 = 0$ and $E \|\tilde{W}_1\|^{\varrho'} < \infty$ for some $\varrho' > \varrho$.
	Finally, assume that 
	$\Sigma_W  = \lim_{n\to\infty} n^{-1} \mathrm{Var}(\sum_{t=1}^n \tilde{W}_t)$
	is positive definite.
	Then on a suitable probability space, one can redefine the sequence  $\{\tilde{W}_i\}_{i=1}^\infty$ together with  $\{C(t)\}_{t=1}^\infty$, where 
	$\{C(t)\}_{t=1}^\infty$ are i.i.d. $d$ dimensional standard normal random vectors, such that 
	such that  with probability 1, for $n$ large enough, 
	\begin{equation}\nonumber
		\left\|\sum_{t=1}^n \tilde{W}_t - \Gamma_{W}B(n) \right\| = o(n^{1/\varrho}),
	\end{equation}
	where $B(n)=\sum_{t=1}^n C(t)$, and $\Gamma_W$ is some $d \times d$ constant matrix satisfying $\Gamma_W \Gamma_W ^\top = \Sigma_W$. 
\end{lemma}

For $n \in \mathbb{N}_+$, let 
$$
\xi(n) = \max\{ i: T_i \leq n\}.
$$
This is the number of regenerations up to time~$n$.
We shall approximate $\sum_{t=1}^n (f(X_t) - \theta)$ by 
\[
\sum_{i=1}^{\xi(n)-1} \tilde{Y}_i = \sum_{t = kT_1+1}^{kT_{\xi(n)}} f(X_t) - k \theta (T_{\xi(n)}-T_1) ,
\]
and apply Lemma~\ref{lem:liu} to $\sum_{i=1}^{\xi(n)-1} \tilde{Y}_i$.
The following lemma will be used to handle the errors of this approximation.

\begin{lemma} \label{lem:error}
	Assume that Condition \ref{con:min} holds. 
	Suppose that $K_U$ is at least polynomially ergodic of order 
	$q>1+\gamma + (2+\gamma)^2/\gamma^*$, 
	and $ \pi (\|f\|^{2+\gamma + \gamma^*} ) <\infty$, for some $\gamma\geq 0$ and $\gamma^*>0$. 
	Then each of the following holds, regardless of the initial distribution of $\{X_t\}_{t=0}^{\infty}$.
	\begin{enumerate}
		\item [(i)] With probability 1,
		\[
		\lim_{n\to\infty} n^{-1/(2+\gamma)} \left\| \sum_{t=1}^{kT_{1}}  f(X_t) \right\| = 0, \text{ and } \lim_{n\to\infty} n^{-1/(2+\gamma)}T_1  = 0.
		\]
		\item [(ii)] For $r \in \{1,\dots,k\}$, with probability 1, 
		\[
		\left\| \sum_{t = kT_{\xi(n)} + 1}^{kn + r} f(X_t) \right\| = O(n^{1/(2+\gamma)}), \text{ as } n\to\infty. 
		\]
		\item [(iii)] With probability 1,
		\[
		 |n - T_{\xi(n)}| = O(n^{1/(2+\gamma)}), \text{ as } n\to\infty.
		\]
		\item [(iv)] For $q' \in [1,2)$, with probability 1,
		\[
		\lim_{n\to\infty} {n}^{-1/q'}(n - \xi(n)\Xi_{\tau}) = 0.
		\]
	\end{enumerate}
\end{lemma}

The proof of Lemma \ref{lem:error} is provided in Appendix \ref{proof:lemerror}. The proof of (ii)-(iv) in Lemma \ref{lem:error} uses some arguments from the proof of \pcite{banerjee2022multivariate} Theorem 1. 
We are now ready to prove Theorem~\ref{thm:sip}.

\begin{proof}[Proof of Theorem \ref{thm:sip}]
Fix $\rho > \max\{1/(2+\gamma), 1/4\}$.

We begin by checking the conditions in Lemma \ref{lem:liu}, taking $\varepsilon_n = \Delta_n$ and $\tilde{W}_n = \tilde{Y}_n$ for $n \in \mathbb{N}_+$.
Clearly, $E\tilde{Y}_1 = 0$. By Lemma \ref{lem:4}, $E\left\|\tilde{Y}_1\right\|^{2+\gamma}<\infty$. 
By Lemma \ref{lem:vfvar}, the matrix $\Sigma_Y$ is positive definite. 
Let $\varrho\in(2, \min\{2 + \gamma, 4\})$ be such that $1/\varrho < \rho$.
By Lemma \ref{lem:liu}, there exists a richer  probability space, on which one can redefine the sequence $\{\tilde{Y}_i\}_{i=1}^\infty$ together with $\{C(t)\}_{t=1}^\infty$, where   
$\{C(t)\}_{t=1}^\infty$ are i.i.d. $d$ dimensional standard  normal random vectors, such that
with probability 1, 
\begin{equation}\label{eqn:sippre}
	\left\| \sum_{i=1}^{m} \tilde{Y}_i - \Gamma_Y B(m) \right\| = o(m^{1/\varrho}), \text{ as } m\to\infty,
\end{equation}
where $B(m)= \sum_{t=1}^{m} C(t)$, and $\Gamma_Y$ is come $d\times d$ constant matrix satisfying $\Gamma_Y\Gamma_Y^\top = \Sigma_Y$.

Following Lemma \ref{lem:measureexist}, stated right after the proof, we may and shall, on a suitable probability space, redefine the sequences $\{f(X_t)\}_{t=1}^\infty$  and $\{\delta_t\}_{t=0}^\infty$ together with $\{B(t)\}_{t=0}^\infty$ and $\{B(t/k\Xi_{\tau})\}_{t=0}^\infty$, so that the following hold:
$\{B(t)\}_{t=0}^\infty$ and $\{B(t/k\Xi_{\tau})\}_{t=0}^\infty$ have the same law as the corresponding elements in a $d$ dimensional standard Brownian motion;
and
with probability 1,  \eqref{eqn:sippre} holds.

By Lemma~\ref{lem:error}-(iii), $\xi(n) \to \infty$ as $n \to \infty$, almost surely.
Then, by \eqref{eqn:sippre}, with probability 1,
$$
\left\| \sum_{i=1}^{\xi(n)-1} \left(\sum_{t=kT_{i}+1}^{kT_{i+1}} f(X_t)\right) - \sum_{i=1}^{\xi(n)-1} \tau_i k \theta - \Gamma_Y B(\xi(n)-1) \right\| = o\left((\xi(n)-1)^{1/\varrho}\right) = o(n^{1/\varrho}).
$$ 
That is, with probability 1, as $n\to\infty$,
\begin{equation}\label{eqn:sm}
	\left\| \sum_{t=kT_{1}+1}^{kT_{\xi(n)}} f(X_t)- k (T_{\xi(n)}-T_1) \theta - \Gamma_Y B(\xi(n)-1) \right\| =  o(n^{1/\varrho}).
\end{equation}

We shall now show, with probability 1, for $r \in \{1,\dots,k\}$, as $n \to \infty$, 
\begin{equation}\label{eqn:sm-1}
	\left\| \sum_{t=1}^{nk+r} f(X_t)- (nk+r) \theta - \Gamma_Y B(\xi(n)-1) \right\| =  o(n^{1/\varrho}).
\end{equation}
Comparing \eqref{eqn:sm} and \eqref{eqn:sm-1}, we see that it suffices to show the following four remainder terms are $o(n^{1/\varrho})$ almost surely: 
$\left\| \sum_{t=1}^{t=kT_{1}}  f(X_t) \right\|$, $T_1$, $\left\| \sum_{t=kT_{\xi(n)}+1}^{kn+r} f(X_t) \right\|$, and $n- T_{\xi(n)}$.
But this is implied by (i)-(iii) of Lemma~\ref{lem:error}.


The final stage of the proof is replacing $B(\xi(n)-1)$ in \eqref{eqn:sm-1} with the Brownian motion evaluated at a non-random time.

Let $q' \in [1,2)$ be such that $1/(2q') < \rho$.
By Lemma~\ref{lem:error}-(iv), with probability 1, for $r \in \{1,\dots,k\}$ and $n$ large enough, 
$$
|\xi(n)-1 - (nk+r)/(k\Xi_{\tau}) | \leq n^{1/q'}.
$$
By \pcite{csorgo2014strong} Theorem 1.2.1 (1.2.4), we have with probability 1, as $n\to\infty$,
$$
\sup_{\{\xi^*:|\xi^* - (nk+r)/(k\Xi_{\tau})|\leq n^{1/q'}\}}  \|B(\xi^*) - B((nk+r)/(k\Xi_{\tau})) \|  = O(\beta_n),
$$
where 
\[
\begin{aligned}
	\beta_n &= \left[2  (n^{1/q'}) \left\{\log \left(\frac{(n+1)/\Xi_{\tau}+n^{1/q'}}{n^{1/q'}}\right)+ \log\log((n+1)/\Xi_{\tau}+n^{1/q'})\right\}\right]^{1/2} \\
	&= O(n^{1/(2q')}\log n).
\end{aligned}
\]
Therefore, with probability 1,
\begin{equation}\label{eqn:bxi}
	\|B(\xi(n)-1) - B((nk+r)/(k\Xi_{\tau})) \|  =  O(n^{1/(2q')}\log n).
\end{equation}

Using the triangle inequality, by \eqref{eqn:sm-1} and \eqref{eqn:bxi}, we have with probability 1, for $r \in \{1,\dots,k\}$,
\begin{equation}\nonumber
	\begin{aligned}
		&\left\|\sum_{t=1}^{kn+r} f(X_t) - (n k+r) \theta -  \Gamma_Y B((nk+r)/(k\Xi_{\tau})) \right\| \\
		&\leq 
		\left\|\sum_{t=1}^{kn+r} f(X_t) - (n k+r) \theta-  \Gamma_Y B(\xi(n)-1)\right\|
		+ \left\|\Gamma_Y \{B(\xi(n)-1) - B((nk+r)/(k\Xi_{\tau})) \}\right\| \\
		&\leq O(n^{1/(2q')}\log n) + o(n^{1/\varrho}).
	\end{aligned}
\end{equation}
Recall that $1/\varrho < \rho$ and $1/(2q') < \rho$.
Then
\begin{equation}\label{eqn:tri}
	\begin{aligned}
		\left\|\sum_{t=1}^{kn+r} f(X_t) - (n k+r) \theta -  \Gamma_Y B((nk+r)/(k\Xi_{\tau})) \right\|  \leq O(n^{\rho}).
	\end{aligned}
\end{equation}

Define  $C'(t) = \sqrt{k \Xi_{\tau}} \{B(t/k\Xi_{\tau}) - B((t-1)/k\Xi_{\tau})\},$ $t\in\mathbb{N}_+$.
Then $\{C'(t) \}_{t=1}^\infty$ are i.i.d. $d$ dimensional standard normal random vectors. 
Let 
$
B'(n) = \sum_{t=1}^{n}C'(t)
$ for $n \in \mathbb{N}_+$.
Since, for $r \in \{1,\dots,k\}$, $$B((nk+r)/(k\Xi_{\tau}))=\frac{1}{\sqrt{k\Xi_{\tau}}} \sum_{t=1}^{nk+r} C'(t) = \frac{1}{\sqrt{k\Xi_{\tau}}} B'(nk+r),$$
by \eqref{eqn:tri}, we have with probability 1,  as $n\to\infty$,
$$
\left\|\sum_{t=1}^{kn+r} f(X_t) - (n k+r) \theta -  \frac{\Gamma_Y}{\sqrt{k \Xi_{\tau}}} B'(nk+r) \right\| \leq O(n^{\rho}),
$$
Therefore, with probability 1, as $n\to\infty$,
\begin{equation}\label{eqn:final}
	\left\|\sum_{t=1}^{n} f(X_t) - n \theta -  \frac{\Gamma_Y}{\sqrt{k \Xi_{\tau}}} B'(n) \right\| \leq O\left(n^{\rho}\right).
\end{equation}
By \eqref{eqn:final} and Theorem \ref{thm:clt}, $\Sigma_Y = k\Xi_{\tau}\Sigma$, so $\Gamma := \Gamma_Y/ \sqrt{k \Xi_{\tau}}$ satisfies $\Gamma \Gamma^{\top} = \Sigma$.
This establishes the desired SIP.
\end{proof}

\begin{lemma}\label{lem:measureexist}
	There exists a probability space, on which one redefine the sequences $\{f(X_t)\}_{t=1}^\infty$ and $\{\delta_i\}_{i=0}^\infty$ together with $\{B(t)\}_{t=0}^\infty$ and $\{B(t/k\Xi_{\tau})\}_{t=0}^\infty$, such that the following holds:
	$\{B(t)\}_{t=0}^\infty$ and $\{B(t/k\Xi_{\tau})\}_{t=0}^\infty$ have the same law as the corresponding elements in a $d$ dimensional standard Brownian motion;
	and
	with probability 1,  \eqref{eqn:sippre} holds. 
\end{lemma}
Lemma \ref{lem:measureexist} can be established in a straightforward manner through the gluing lemma. See Appendix \ref{proof:measure} for details.

\subsection{Strong consistency of the multivariate batch means estimator}\label{ssec:bm}

For time in-homogeneous cyclic Markov chains, we describe the multivariate batch means estimator, and give conditions  under which it is a strongly consistent estimator of the asymptotic covariance matrix $\Sigma$ with the help of the SIP established in Section \ref{ssec:sip}.

Let $n$ represent the length of the in-homogeneous cyclic Markov chain of interest. 
Suppose that $n$ can be written into the product of two positive integers, $a_n$ and $b_n$.
We can partition the chain into $a_n$ number of batches, and each batch contains $b_n$ consecutive elements.
We call $b_n$ the batch length.
For $i \in \{1,\dots , a_n\}$, let 
$$\hat{\theta}_n^{(i)} = \frac{\sum_{t=(i-1)b_n + 1}^{ib_n} f(X_t)}{b_n}$$ 
be the sample mean of the $i$th batch, and recall that $\hat{\theta}_n = \sum_{t=1}^{n} f(X_t)/n$ is the sample mean of the whole chain. 
Since $f(X_t)\in\mathbb{R}^d$, $\hat{\theta}_n^{(i)}$, $i \in \{1,\dots , a_n\}$,  and $\hat{\theta}_n$ are $d$ dimensional random vectors.
Define a multivariate batch means estimator as
\begin{equation}\label{eqn:bm}
\hat{\Sigma}^{\text{BM}}_{n} = \frac{b_n}{a_n-1} \sum_{i=1}^{a_n}(\hat{\theta}_n^{(i)} - \hat{\theta}_n)(\hat{\theta}_n^{(i)} - \hat{\theta}_n)^\top.
\end{equation}
For time homogeneous samplers, $\hat{\Sigma}^{\text{BM}}_{n}$ is a commonly used estimator for the asymptotic variance of $\hat{\theta}_n$.


The main contribution of this section is the following theorem, which states that, for in-homogeneous cyclic chains, $\hat{\Sigma}^{\text{BM}}_{n}$ is consistent for estimating the asymptotic covariance matrix~$\Sigma$.

 \begin{theorem}\label{thm:strongbm}
	Suppose that the assumptions in Theorem \ref{thm:sip} hold.
	If the batch length $b_n$ satisfies Conditions \ref{con:strongbm} and \ref{con:strongbm2} below, then with probability 1, $\hat{\Sigma}^{\text{BM}}_{n}\to\Sigma $ as $n\to\infty$.
 \end{theorem}

\begin{condition}\label{con:strongbm}
	Each of the following holds.
	\begin{enumerate}[a.]
		\item \label{strongbma} The number of batches sequence $\{a_n\}_{n=1}^\infty$ is monotonically increasing, and $a_n \to \infty$ as $n \to \infty$. 
		\item \label{strongbmb} The batch length sequence $\{b_n\}_{n=1}^\infty$ is monotonically increasing, 
		and $b_{n}\to\infty$ as $n\to\infty$. 
		\item \label{strongbmc} There exists $c\geq1 $ such that $\sum_{t=1}^\infty (1/a_t)^c < \infty$.
	\end{enumerate}
\end{condition}

\begin{condition}\label{con:strongbm2}
	There exists a number $\rho > \max\{1/(2+\gamma), 1/4\}$, where $\gamma$ is defined in Theorem~\ref{thm:sip}, such that
	$b_n^{-1}n^{2 \rho} \log n \to 0$ as $n\to \infty$.
\end{condition}

\begin{remark}
	One natural choice of the batch length $b_n$ is $\lfloor n^{\kappa}\rfloor$, $\kappa\in(0,1)$,  which satisfies Condition \ref{con:strongbm}. In order to satisfy Condition \ref{con:strongbm2}, pick $\kappa>\max\{2/(2+\gamma), 1/2\}$.  
\end{remark}

Theorem~\ref{thm:strongbm} follows from the SIP in Theorem \ref{thm:sip} herein and Theorem 4 in \pcite{vats2019multivariate} Supplement.
\pcite{vats2019multivariate} Theorem 4 states that under regularity conditions, the SIP implies the consistency of the batch means estimator.
Although \cite{vats2019multivariate} focused on time homogeneous chains, their proof for that theorem remains valid for in-homogeneous chains.





\subsection{A multivariate sequential termination rule}\label{ssec:tr}

For time homogeneous Markov chains, \cite{vats2019multivariate} developed an asymptotically valid multivariate fixed volume sequential termination rule. 
Using the SIP and the multivariate batch means estimator developed in previous sections, we adopt their approach to time in-homogeneous cyclic Markov chains. We specify conditions under which the multivariate fixed volume sequential termination rule can  achieve fixed volume confidence region asymptotically.

Let $\alpha \in (0,1)$ be a constant.
Define $T^2_{1-\alpha, d, p}$ to be the $(1-\alpha)$th quantile of a Hotelling's $t$-squared distribution  with dimension parameter $d$ and degrees of freedom $p$.
Recall that $\hat{\Sigma}_n^{\text{BM}}$ defined in   \eqref{eqn:bm} is the multivariate batch means estimator of $\Sigma$. 
Assume that $\Sigma$ is positive definite, so almost surely, $\hat{\Sigma}_n^{\text{BM}}$ is positive definite for large $n$.
An asymptotic $(1-\alpha)$ confidence region for $\theta $ can be defined as
\begin{equation}\label{eqn:region}
S_{\alpha}(n) = \{x\in \mathbb{R}^d: n(\hat{\theta}_n - x)^\top (\hat{\Sigma}^{\text{BM}}_n)^{-1} (\hat{\theta}_n - x)< T^2_{1-\alpha, d, p_n}\},
\end{equation}
where the degrees of freedom $p_n = a_n-d$. 
Then $S_{\alpha}(n)$ is a $d$ dimensional ellipsoid with axes along the direction of the eigenvectors  of $\hat{\Sigma}^{\text{BM}}_n$.
The volume of $S_{\alpha}(n)$ is 
$$V(S_{\alpha}(n)) = \frac{2\pi^{d/2}}{dG(d/2)}\left(\frac{T^2_{1-\alpha, d, p_n}}{n}\right)^{d/2}|\hat{\Sigma}_n^{\text{BM}}|^{1/2},$$
where $G(\cdot)$ is the Gamma function.

Let the volume parameter $\varepsilon$ be a small positive value.
Let $\Psi = \text{Var}_{\pi} f$, the covariance matrix of $f(X)$, where $X\sim\pi$, and let $\hat{\Psi}_n$ be the usual sample covariance estimator of  $\Psi$. 
Let $\mathfrak{M}(f,\pi)>0$ be  
a number that is determined by $f$ and $\pi$,
and let $\hat{\mathfrak{M}}_n(f,\pi)>0$ be a strongly consistent estimator of $\mathfrak{M}(f,\pi)$. 
Some choices of $\mathfrak{M}(f,\pi)$ are: $\mathfrak{M}(f,\pi)=1$ and $\hat{\mathfrak{M}}_n(f,\pi)=1$, or $\mathfrak{M}(f,\pi)= |\Psi|^{1/(2d)}$ and $\hat{\mathfrak{M}}_n(f,\pi)=|\hat{\Psi}_n|^{1/(2d)}$. 
Let $n_0\in\mathbb{N}_+$, and let $\mathbb{I}(x<n_0)$ be the indicator function with $\mathbb{I}(x<n_0)=1$ if $x<n_0$, and $\mathbb{I}(x<n_0)=0$ otherwise.
Consider terminating the MCMC simulation at time
\begin{equation} \label{eqn:ter}
	\tN(\varepsilon) = \inf\{n\in\mathbb{N}: \{V(S_{\alpha}(n))\}^{1/d} +s(n, \varepsilon) \leq  \varepsilon \hat{\mathfrak{M}}_n(f,\pi)\},
\end{equation}
where $s(n, \varepsilon) = \varepsilon \hat{\mathfrak{M}}_n(f,\pi) \mathbb{I}(n<n_0) + n^{-1}$.
It is easy to show that, almost surely, $N(\varepsilon) < \infty$.
The term $s(n, \varepsilon)$ guarantees that the MCMC simulation executes at least $n_0$ iterations, preventing early stopping. 
Terminating the simulation at time $\tN(\varepsilon)$ roughly equates terminating when the volume of the $1-\alpha$ confidence region $S_{\alpha}(n)$ drops below $\varepsilon \mathfrak{M}_n(f, \pi)$.

When $\mathfrak{M}(f,\pi)= |\Psi|^{1/(2d)}$, the termination rule has a close connection to the concept of effective sample sizes (ESS).
Define the multivariate ESS 
\begin{equation}\label{eqn:ess}
\text{ESS} = n\left(\frac{|\Psi|}{|\Sigma|}  \right)^{1/d},
 \end{equation}
and the estimated ESS
$$\widehat{\text{ESS}}_n = n\left(\frac{|\hat{\Psi}_n|}{|\hat{\Sigma}_n^{\text{BM}}|}  \right)^{1/d},$$
and assume that $\hat{\Sigma}_n^{\text{BM}}$ and $\hat{\Psi}_n$ are strongly consistent estimators of $\Sigma$ and $\Psi$ respectively. 
Then,  with probability 1, 
$$\frac{\widehat{\text{ESS}}_n}{\text{ESS}}\to 1, \text{ as } n\to\infty.$$  
When $\mathfrak{M}(f,\pi)= |\Psi|^{1/(2d)}$, for $n>n_0$ and large enough, the sequential termination rule   $n \geq \tN(\varepsilon)$ can be rearranged as
$$
		\widehat{\text{ESS}}_n \geq   \left\{\left(\frac{2\pi^{d/2}}{dG(d/2)}\right)^{1/d} (T^2_{1-\alpha, d, p_n})^{1/2}  + \frac{1}{|\hat{\Sigma}_n^{\text{BM}}|^{1/(2d)} n^{1/2}}\right\}^2 \frac{1}{\varepsilon^2}  \approx  \left(\frac{2\pi^{d/2}}{dG(d/2)}\right)^{2/d} \frac{T^2_{1-\alpha, d, p_n}}{\varepsilon^2} 
$$
Define $\chi^2_{1-\alpha, d}$ to be the $(1-\alpha)$th quantile of a chi-square distribution with degrees of freedom $d$.
When $n\to\infty$, $p_n\to\infty$, and $T^2_{1-\alpha, d, p_n}\to \chi^2_{1-\alpha, d}$. Therefore, when $n$ is sufficiently large, the termination rule roughly becomes
\begin{equation}\nonumber
\widehat{\text{ESS}}_n  \geq  \left(\frac{2\pi^{d/2}}{dG(d/2)}\right)^{2/d}\frac{\chi^2_{1-\alpha, d}}{\varepsilon^2}.
\end{equation}
By this display, one can a priori determine the minimum ESS required to obtain a $(1-\alpha)$ confidence region with volume $\varepsilon |\Psi|^{1/(2d)}$.


The theorem below shows that the fixed volume sequential termination rule in   \eqref{eqn:ter} is asymptotically valid.

\begin{theorem}\label{thm:ter}
	Suppose that the assumptions in Theorem \ref{thm:sip} hold, and that $\hat{\mathfrak{M}}_n(f,\pi)$ is a strongly consistent estimator of $\mathfrak{M}(f,\pi)$.
	If the batch length $b_n$ satisfies Conditions \ref{con:strongbm} and \ref{con:strongbm2}, 
	then, for $\alpha \in (0,1)$, as $\varepsilon\to 0$, $\tN(\varepsilon)\to\infty$ and $P(\theta \in S_{\alpha}(\tN(\varepsilon)))\to 1-\alpha$.
\end{theorem}

\begin{remark}
Consider $\mathfrak{M}(f,\pi)=|{\Psi}|^{1/(2d)}$ and $\hat{\mathfrak{M}}_n(f,\pi)=|\hat{\Psi}_n|^{1/(2d)}$. 
Then $\hat{\mathfrak{M}}_n(f,\pi)$ is a strongly consistent estimator of $\mathfrak{M}(f,\pi)$ by Theorem \ref{thm:slln} and the continuous mapping theorem. 
\end{remark}

The proof of Theorem~\ref{thm:ter}, given in Appendix \ref{proof:thmter},  relies on the SIP in Theorem \ref{thm:sip} as well as the consistency of $\hat{\Sigma}_n^{\text{BM}}$ established in Theorem \ref{thm:strongbm}.
Part of its proof uses some arguments from earlier works on sequential stopping rules \citep{glynn1992asymptotic, vats2019multivariate}.


\section{Numerical Experiments}\label{sec:num}


With our uncertainty assessment tools, we are able to perform output analysis for time in-homogeneous cyclic MCMC samplers. 
We use numerical examples to illustrate the methods described in Sections \ref{ssec:bm} and \ref{ssec:tr}.
In each of the following numerical examples, we specify a target distribution $\pi$, a function $f$, and estimate $\pi (f)$ using MCMC samplers. 
We compare the efficiency of two samplers:  a time in-homogeneous cyclic MCMC  sampler with stationary distribution $\pi$ and a natural homogeneous counterpart of that sampler.
Each example comprises two experiments: a fixed length experiment, and a termination rule experiment. 
The experiments are replicated multiple times.

In the fixed length experiment, we employ the two samplers to generate two Markov chains with the same fixed length. 
For the two chains, we calculate the ESS  and record the computation time to generate them. 
We compare the in-homogeneous chain and the homogeneous chain in terms of their ESS per minute (ESSpm).
When the function $f$ is multivariate, the ESS defined in \eqref{eqn:ess} is based on determinants, tracking the  products of the eigenvalues of $\Psi$ and $\Sigma$, respectively. 
In one of our examples, we will also assess sampler performance by estimating the traces of $\Psi$ and $\Sigma$, which correspond to the sums of the eigenvalues. 
Define the trace effective sample size (TESS) as
\begin{equation}\label{eqn:tess}
	\text{TESS} = n \frac{\text{tr}(\Psi)}{\text{tr}(\Sigma)},
\end{equation} 
which is estimated by
$$
\widehat{\text{TESS}}_n = n \frac{\text{tr}(\hat{\Psi}_n)}{\text{tr}(\hat{\Sigma}^{\text{BM}}_n)}.
$$
We compare the in-homogeneous chain and the homogeneous chain based on their TESS per minute (TESSpm).
The  asymptotic covariance matrices are estimated by the batch means approach with batch length $b_n \approx n^{0.51}$.
We create $90\%$ confidence regions for $\theta = \pi(f)$ via~\eqref{eqn:region}.
We use repeated experiments to find the empirical coverage rates of the confidence regions.


In the termination rule experiment, we execute both the  time in-homogeneous cyclic MCMC  sampler and its corresponding time homogeneous sampler under the same fixed volume sequential termination rule. 
The computation costs up to termination are then compared.
We also construct $90\%$ confidence regions 
for $\theta$ at the time of termination via \eqref{eqn:region}, and we compute their empirical coverage rates through repeated experiments.
The sequential termination rule is based on $90\%$ confidence regions with volume parameter $\varepsilon=0.05$. 
We check the rule every $20\%$ increments of iterations.

\subsection{Bivariate uniform models}

Let $h(\cdot)$ be a real valued function on $\mathbb{R}$. 
Consider the area under the curve $S = \{(x_1,x_2)\in\mathbb{R}\times\mathbb{R}:0\leq x_2\leq h(x_1)\}$.
Let $\pi$ be the uniform distribution on the area $S$, and let $f(x_1,x_2) = x_1\times x_2$ for $(x_1,x_2)\in\mathbb{R}^2$. 
We aim to estimate  $\pi (f)$. 
The conditional distribution of $X_2$ given $X_1=x_1$ is uniform on internal $[0,h(x_1)]$. Given $(x_1, x_2) \in S$, a y-axis step draws a point $x'_2$ from the conditional distribution of $X_2$ given $X_1=x_1$ and outputs $(x_1, x'_2)$. 
The conditional distribution of $X_1$ given $X_2=x_2$ is uniform on region $\{x_1\in\mathbb{R}:  h(x_1)\geq x_2\}$. Given $(x_1, x_2) \in S$, an  x-axis step draws a point $x'_1$ from the conditional distribution of $X_1$ given $X_2=x_2$ and outputs $(x'_1, x_2)$. 
Alternating between the two steps constitutes a two-component Gibbs sampler. 
Implementing the x-axis step involves solving $h(x_1)=x_2$, and is apparently more difficult than the y-axis step. 
To achieve smaller variance under a given computation time,   one can use a modified deterministic scan Gibbs sampler that conducts the  y-axis step $k-1$ times, where $k$ may be greater than 2, before conducting the x-axis step once \citep{qin2022analysis}.
As noted in Section \ref{sssec:sip}, this is a time in-homogeneous cyclic sampler. 
For this experiment, we take $k=4$.
We also consider a time homogeneous counterpart, which performs a y-axis step followed by an x-axis step in each iteration.
Note that the chain associated with the homogeneous sampler has the same transition law as the first homogeneous subchain of the in-homogeneous sampler.

Let $h(x) = 2x+1-e^x$. 
The conditional distribution of $X_1$ given $X_2=x_2$ 
is uniform on region  
$\{x_1\in\mathbb{R}:  h(x_1)\geq x_2\} = [r_l, r_r]$,
where
$r_l$ and $r_r$ are the left and right roots for $h(x_1) - x_2=0$. The roots are found using the \texttt{uniroot} function in the \texttt{stats} R package~\citep{rcore}.

By the discussion in Section \ref{sssec:sip}, Condition \ref{con:min} holds.
In addition, by~\pcite{roberts2004general} Theorem 8, the transition law of the first  homogeneous subchain is geometrically ergodic. 
By Remark \ref{rem:ergodic}, $K_U$ defined in \eqref{eqn:KU} is geometrically ergodic.
This guarantees that Theorems \ref{thm:sip}, \ref{thm:strongbm}, and \ref{thm:ter} all apply to the two samplers being investigated.

We carry out the fixed length and termination experiments described at the beginning of this section, and we repeat the process $1000$ times.


In the fixed length segment, we run each sampler for $3 \times 10^4$ iterations.
The results from 1000 repeated experiments are shown in Table \ref{table:11} with standard errors in parenthesis.

\begin{table}[!htbp]
	\centering
	\caption{
		Bivariate uniform: fixed length experiment.
		The Sampler column labels the samplers used; the Time column stores the computation time in seconds; the ESS and ESSpm columns record the ESS and ESS per minute, respectively; the Coverage column stores the empirical coverage probability of the 90\% confidence region.
	}
	\label{table:11}
	\begin{tabular}{lllll}
		\hline
		Sampler & Time & ESS &ESSpm & Coverage  \\ \hline
		In-homo & 33.11 (0.06) & 9885 (36)  & 17963 (72) & 0.895 (0.010) \\
		Homo & 123.50 (0.28) & 22694 (83)   & 11071 (46) & 0.895 (0.010)\\ \hline
	\end{tabular}
\end{table}

From Table \ref{table:11}, 
the ESSpm of the in-homogeneous sampler is 62.3\% greater than that of the homogeneous sampler.
The empirical coverage probability is close to 90\%, suggesting that our estimated asymptotic
 variance in the CLT is accurate.
After taking both the ESS and the computation time into account, the  modified deterministic scan Gibbs sampler is more efficient than its time homogeneous counterpart. This highlights the usefulness of tools for the output analysis of in-homogeneous cyclic chains.

Table \ref{table:13} gives the result of 1000 repetitions of the termination rule experiment.

\begin{table}[!htbp]
	\centering
	\caption{
	Bivariate uniform: termination rule experiment.
	The Sampler column indicates samplers used; the Time column records the computation time in seconds at termination; the Iter (x-axis) column and Iter (y-axis) column store the numbers of x-axis and y-axis steps used, respectively; the Coverage column stores the empirical coverage probability of the 90\% confidence region. 
	}
	\label{table:13}
\begin{tabular}{lllll}
	\hline
	Samplers & Time  & Iter (x-axis) & Iter (y-axis) & Coverage \\ \hline
	In-homo   & 17.97 (0.09) & 3815 (17) & 11444 (51)  &  0.909 (0.009) \\
	Homo   & 31.25 (0.21) & 6970 (34) & 6970 (34) & 0.891 (0.010)  \\ \hline
\end{tabular}
\end{table}

From Table \ref{table:13}, 
under the same termination rule,
the computation time of the in-homogeneous sampler is 42.5\% shorter than that of the homogeneous sampler. 
Compared with the homogeneous sampler, the modified deterministic scan Gibbs sampler uses much less number of costly x-axis iterations and implements more y-axis steps. 
The empirical coverage probabilities are close to 90\%, which is consistent with our Theorem \ref{thm:ter}.

\subsection{Bayesian linear mixed models}

We consider a Bayesian linear mixed model with a proper prior  and a Gibbs sampler described in \cite{roman2015geometric} targeting the posterior.

Let $g$ be the number of subjects, and suppose that, in each subject $i$, $i\in\{1,\dots, g\}$, there are $n_i$ observations.
Consider the following  linear mixed model
$$\tilde{Y}_{i,j} = x_{i,j}^{\top} \beta +  \gamma_i + \epsilon_{i,j}, $$
where for $i\in\{1,\dots, g\}$ and 
$j\in\{1,\dots,n_i\}$, 
the feature $x_{i,j}$s are $p\times 1$ known  vectors, the fixed effects $\beta$ is a $p\times 1$ vector,  the $\gamma_i$s are i.i.d. $\mathcal{N}(0,\lambda_{\gamma}^{-1})$, and the $\epsilon_{ij}$s are i.i.d. $\mathcal{N}(0,\lambda_{e}^{-1})$  and independent of the $\gamma_i$s. The random variables $\gamma_1, \dots, \gamma_g$ are usually called subject specific intercepts and the $\epsilon_{i,j}$s are within subject errors.
Let  $\tilde{n}=\sum_{i=1}^g n_i$ be  the total number of observations.
Let $Y$ be an $ \tilde{n} \times 1$ vector obtained by stacking all responses $\tilde{Y}_{i,j}$s so that the $b$th element of $Y$ is $\tilde{Y}_{i,j}$ if $b=\sum_{l=1}^i n_l+ j$. 
Similarly, let $X$ be a $\tilde{n}\times p$ known feature matrix obtained by stacking all  $x_{i,j}^{\top}$s so that the $b$th row of $X$ is $x^{\top}_{i,j}$ if $b=\sum_{l=1}^i n_l+ j$, and let $E$ be an $\tilde{n}\times 1$ vector obtained by stacking all errors $\epsilon_{ij}$s so that the $b$th element of $E$ is $\epsilon_{i,j}$ if $b=\sum_{l=1}^i n_l+ j$.
Define $\gamma=(\gamma_1,\dots,\gamma_g)^{\top}$, and let $Z={(Z_{s,t})_{s=1}^{\tilde{n}}}_{t=1}^g$ be an $\tilde{n}\times g$ matrix. 
For $t\in\{1,\dots,g\}$, $Z_{s,t}=1$ if $s\in\{\sum_{i=1}^{t-1}n_i +1,\dots, \sum_{i=1}^{t}n_i\}$  and $Z_{s,t}=0$ otherwise. 
In other words, $Z_{s,t} = 1$ if and only if, after the stacking, the $s$th observation belongs to the $t$th subject.
Then the linear mixed model can be summarized as
$$Y = X\beta + Z\gamma + E.$$

The Bayesian linear mixed model with proper priors can be described in the following hierarchical form. The first stage is 
$$Y|\beta,\gamma,\lambda_e \sim \mathcal{N}_{\tilde{n}}(X\beta + Z\gamma, \lambda_{e}^{-1}I_{\tilde{n}\times \tilde{n}}),$$
where $I_{\tilde{n}\times \tilde{n}}$ is $\tilde{n}\times \tilde{n}$ identity matrix and $\mathcal{N}_{\tilde{n}}$ represents $\tilde{n}$ dimensional multivariate normal distribution. The second stage is
$$\gamma|\lambda_\gamma, \lambda_e \sim \mathcal{N}_{g}(0, \lambda_{\gamma}^{-1}I_{g\times g}), \ \text{and } 
\beta|\lambda_\gamma, \lambda_e \sim \mathcal{N}_{p}(\mu_{\beta}, \Sigma_{\beta}),
$$
where given $\lambda_\gamma$ and $\lambda_e$, $\gamma$ and $\beta$ are mutually independent.
Denote a Gamma distribution with shape parameter $a$ and rate parameter $b$ by Gamma$(a,b)$. The density of Gamma$(a,b)$ is 
$f_G(x;a,b)\propto x^{a-1}e^{-bx},$ $x>0$. The third stage is
$$
\lambda_{\gamma}\sim \text{Gamma}(a_{\gamma}, b_{\gamma}) \text{ and } \lambda_{e}\sim \text{Gamma}(a_e,b_e), 
$$
where $\lambda_{\gamma}$ and $\lambda_{e}$ are mutually independent. The hyperparameters $\mu_{\beta}, \Sigma_{\beta}, a_{\gamma}, b_{\gamma}, a_e,$ and $b_e$ are known, where $\Sigma_{\beta}$ is positive definite, and $a_{\gamma}, b_{\gamma}, a_e,$ and $b_e$ are positive real values.
Let $f_\mathcal{N}(x;a,B)$ denote the density of $\mathcal{N}(a,B)$. Given $Y=y$, the posterior $\pi(\cdot|y)$ is
$$
\begin{aligned}
	\pi(\beta, \gamma, \lambda_\gamma, \lambda_e|y) \propto & f_\mathcal{N}(y; X\beta + Z\gamma, \lambda_{e}^{-1}I_{n\times n}) f_\mathcal{N}(\gamma; 0,  \lambda_{\gamma}^{-1}I_{g\times g}) f_\mathcal{N}(\beta; \mu_{\beta}, \Sigma_{\beta})  \\
	& f_G(\lambda_{\gamma}; a_{\gamma}, b_{\gamma}) f_G(\lambda_{e}; a_e,b_e).
\end{aligned}
$$

The posterior is intractable in the sense that it is hard to analytically calculate its expectation. This motivates the use of Gibbs samplers targeting the posterior. 

Given $(\beta, \gamma)$ and $Y=y$, $\lambda_\gamma $ and $\lambda_e$ are independent with
\begin{equation}\label{eqn:lambda}
\begin{aligned}
	& \lambda_\gamma| (\beta, \gamma), y \sim \text{Gamma} (a_\gamma + g/2, b_\gamma + \|\gamma\|^2/2), \\
	& \lambda_e| (\beta, \gamma), y \sim \text{Gamma} (a_e + \tilde{n}/2, b_e + \|y - X\beta + Z\gamma\|^2/2),
\end{aligned}
\end{equation}
 where $\|\cdot\|$ is the Euclidean norm of a vector. 
 We refer to the process of sampling from it as the lambda step.
 
Given $\lambda_\gamma $, $\lambda_e$, and $Y=y$, $(\beta, \gamma)$ is a multivariate normal distribution.
Define $A_\lambda = (\lambda_e X^\top X + \Sigma_\beta^{-1})^{-1}$ , $B_\lambda = I_{\tilde{n}\times\tilde{n}} - \lambda_e X A_\lambda X^{\top}$, and $C_\lambda = (\lambda_e Z^\top B_\lambda Z + \lambda_\gamma I_{g\times g})^{-1}$. 
The mean of the multivariate normal distribution is
$$\mu((\beta, \gamma)|\lambda_\gamma, \lambda_e, y) = 
\begin{pmatrix}
	A_\lambda(\lambda_eX^\top y + \Sigma_\beta^{-1}\mu_\beta) - \lambda_e^2 A_\lambda X^\top Z C_\lambda Z^\top (B_\lambda y - XA_\lambda \Sigma_\beta^{-1} \mu_\beta)  \\
	\lambda_e C_\lambda Z^\top (B_\lambda y - XA_\lambda\Sigma_{\beta}^{-1} \mu_\beta)  
\end{pmatrix},$$
and the covariance matrix is
$$
\Sigma((\beta, \gamma)|\lambda_\gamma, \lambda_e, y) =
\begin{pmatrix}
	A_\lambda + \lambda_e^2 A_\lambda X^\top Z C_\lambda Z^\top X A_\lambda & -\lambda_e A_\lambda X^\top Z C_\lambda \\
	 -\lambda_e C_\lambda Z^\top X A_\lambda   & C_\lambda 
\end{pmatrix}.
$$ We call sampling from it beta-gamma step.

Alternating between the two steps constitutes a two-component Gibbs sampler, and the beta-gamma  step requires significantly more computational resources compared to the lambda step.  To achieve smaller Monte Carlo variance and computation time, we use a modified deterministic scan Gibbs sampler suggested by \cite{qin2022analysis}. In this experiment, we update the lambda step three times before updating the  beta-gamma step once, i.e., $k=4$. The modified deterministic scan Gibbs sampler is thus a time in-homogeneous 4-cyclic MCMC sampler. 
Denote the underlying chain by $\{(\beta_t, \gamma_t, \lambda_{\gamma t}, \lambda_{et})\}_{t=0}^\infty$.
This chain will be compared with one of its homogeneous counterpart, which has the same transition law as
 its first homogeneous subchain $\{(\beta_{4t}, \gamma_{4t}, \lambda_{\gamma, 4t}, \lambda_{e, 4t})\}_{t=0}^\infty$.

By Section \ref{sssec:sip}, Condition \ref{con:min} holds. 
Assume that that $X$ has full column rank, $g \geq 2$, and that there exist $i,j \in \{1,\dots,g\}$ such that $i \neq j$ and $n_i \geq 2$, $n_j \geq 2$.
Then the transition kernel of the first homogeneous subchain is geometrically ergodic by the proposition below. 
Thus, the kernel $K_U$ defined in \eqref{eqn:KU} is geometrically ergodic by Remark~\ref{rem:ergodic}.

\begin{proposition} \label{pro:mixedgeo}
	The transition kernel of the first homogeneous subchain is geometrically ergodic if the feature matrix
	$X$ has full column rank, the number of subjects $g\geq 2$, and there exist at least two subjects $i,j \in\{1,\dots,g\}$, such that the number of observations within the $i$th subject are larger than 1, i.e., $n_i \geq 2$ and $n_j \geq 2$.
\end{proposition}

\begin{proof}
By a de-initialzing argument, the first homogeneous subchain has the same convergence properties as the underlying Markov chain of the data augmentaiton algorithm that~\cite{roman2015geometric} studied.
Then we can show geometric ergodicity by checking the conditions in \pcite{roman2015geometric} Proposition 1. 
\end{proof}

Consider the Orthodont data from \cite{potthoff1964generalized}. The investigators followed the change in an orthdontic measurement every 2 years from age 8 until age 14 for several young subjects. There are $g=27$ subjects including 16 boys and 11 girls.  The response is the distance from the pituitary to the pterygomaxillary fissure (in mm) measured on X-ray images of the skull. The features include the age and sex of the subject. 
We use the Bayesian linear mixed model with fixed effects including  the intercept, age effect, and sex effect. Set the hyperparameters as follows:  $(a_\gamma, b_\gamma) = (1,1)$ for $\lambda_\gamma$, $(a_\epsilon, b_\epsilon) = (1,1)$ for $\lambda_\epsilon$, and $\mu_\beta = (0,0,0)$ and $\Sigma_\beta = I_{3\times 3}$ for $\beta$.
The slope parameter for gender and the precision of the random intercept  $(\beta_{\text{Male}}, \lambda_\gamma)$  is the function $f$ of interest.
We use the modified deterministic scan Gibbs sampler and its homogeneous counterpart to sample from the posterior. 
It can be checked that, for this dataset, the conditions of Proposition~\ref{pro:mixedgeo} are satisfied, so Theorems~\ref{thm:sip},~\ref{thm:strongbm}, and~\ref{thm:ter} all apply to the two samplers.
The experiment process, which  is described in the beginning of Section \ref{sec:num}, is repeated $400$ times.


In the fixed length segment, we run each sampler for $16000$ iterations.
The results from 400 repeated experiments are given in Table \ref{table:21}.
In order to estimate the empirical coverage rate, we need to know the true value $\pi(f)$, which is approximated using an exceedingly long chain with a length of $3\times10^6$.

\begin{table}[!htbp]
	\centering
	\caption{Bayesian linear mixed model: 
		fixed length experiment.
		The Sampler column specifies the samplers used, with Inhomo and Homo representing In-homogeneous and Homogeneous, respectively; the Time column records the computation time in seconds; the ESS, ESSpm, TESS, and TESSpm columns record the ESS, ESS per minute, TESS and TESS per minute, respectively; the Coverage column stores the empirical coverage probability of the 90\% confidence region.
	}
	\label{table:21}
	\begin{tabular}{lllllll}
		\hline
		Sampler & Time & ESS & ESSpm & TESS & TESSpm & Coverage  \\ \hline
		Inhomo & 85.40 (0.05) & 6056 (28) & 4255 (20)  & 9062 (57) & 6368 (41)  &  0.890 (0.016) \\
		Homo  & 350.09 (0.77) & 15218 (67) & 2613 (13) & 14930 (91) & 2563 (17)  &   0.898 (0.015) \\ \hline
	\end{tabular}
\end{table}

From Table \ref{table:21}, 
the ESSpm of the in-homogeneous sampler is 62.8\% greater than that of the homogeneous sampler, and the TESSpm of the in-homogeneous sampler is 148.5\% greater than that of the homogeneous sampler.
The empirical coverage probability, closely approximating 90\%, indicates the accuracy of our estimated covariance matrix in the CLT.
Considering both the Monte Carlo estimator performance and the computation time, the modified deterministic scan Gibbs sampler proves to be more efficient than its time homogeneous counterpart. 
This again shows the importance of having specialized output analysis tools for time in-homogeneous samplers.

Table \ref{table:23} gives the result of 400 repetitions of the termination rule experiment.

\begin{table}[!htbp]
	\centering
	\caption{Bayesian linear mixed model: 
		termination rule experiment.
		The Sampler column indicates samplers used;  the Time column records computation time in seconds until terminate; the Iter (beta-gamma) column and Iter (lambda) column store the numbers of beta-gamma and lambda steps used, respectively;
		the Coverage column stores the empirical coverage probability of the 90\% confidence region. 
	}
	\label{table:23}
	\begin{tabular}{lllll}
		\hline
		Samplers & Time  & Iter (beta-gamma)  & Iter (lambda) & Coverage \\ \hline
		In-homo   & 91.51 (0.55) & 4364 (23) & 13093 (70) & 0.890 (0.016)  \\
		Homo & 152.09 (0.70) & 7249 (30) & 7249 (30) & 0.910 (0.014) \\ \hline
	\end{tabular}
\end{table}

From Table \ref{table:23}, 
the computation time of the in-homogeneous sampler is 39.8\% shorter than that of the homogeneous sampler. 
This is because the number of computationally expensive beta-gamma steps iterated in the in-homogeneous sampler is much smaller than in the homogeneous sampler. 
The empirical coverage probability is close to 90\%, which is consistent with our Theorem \ref{thm:ter}.

\section*{Acknowledgments}

The authors thank Prof. Galin Jones for introducing the general SIP topic, and for his valuable suggestions regarding the exact form of the covariance matrix in the CLT and the design of numerical experiments for sampler comparison.
The first and second authors were supported by the National Science Foundation [grant number DMS-2112887].

\newpage
\vspace{1cm}
\noindent{\bf\LARGE Appendix}
\appendix
\section{Proofs for Section \ref{ssec:sllnclt}}



\subsection{Proof of Theorem \ref{thm:slln}}\label{proof:thmslln}


By Lemma~\ref{lem:1}, the transition kernel of the $i$th homogeneous subchain $\{Z_t^i\}_{t=0}^\infty$, $i\in\{1,\dots,k\}$, is Harris ergodic. 
Applying the SLLN for homogeneous chains \citep[][Theorem 17.0.1]{meyn2005markov} to each subchain, we have with probability 1, for $i\in\{1,\dots, k\}$, 
\[
\hat{\theta}_{m,i} := \sum_{t=0}^{m-1} \frac{f(Z_t^i)}{m} = \sum_{t=0}^{m-1} \frac{f(X_{kt+i-1})}{m} \to \theta, \text{ as } m\to\infty.
\]
Then, with probability 1, for $r \in \{1,\dots,k\}$,
\[
\hat{\theta}_{km + r - 1} = \frac{m+1}{km+r-1} \sum_{i=1}^r \hat{\theta}_{m+1,i} - \frac{f(X_0)}{km+r-1} + \frac{m}{km+r-1} \sum_{i=r+1}^k \hat{\theta}_{m,i} \to \theta,  \text{ as } m\to\infty.
\]
That is, with probability 1, for any $\epsilon > 0$ and $r \in \{1,\dots,k\}$, one can find a (random) integer $M_{\epsilon,r}$ large enough so that $\|\hat{\theta}_{km+r-1} - \theta\| < \epsilon$ whenever $m \geq M_{\epsilon,r}$.
Take $M_{\epsilon} = \max_r M_{\epsilon,r}$.
Then, with probability 1, for $\epsilon > 0$, $\|\hat{\theta}_n - \theta\| < \epsilon$ whenever $n/k \geq M_{\epsilon}$.
Thus, $\hat{\theta}_n \to \theta$ almost surely.


\subsection{Proof of Theorem \ref{thm:clt}}\label{proof:clt}

We begin by stating a CLT for homogeneous chains.

\begin{lemma}\label{lem:homoclt}
Let $\{Y_t\}_{t=0}^\infty$ be a time homogeneous Markov chain with state space $(\mathcal{Y}, \mathcal{F}^*)$ and stationary distribution $\pi^*$, starting from an arbitrary initial distribution.
Let $f:\mathcal{Y}\to R^d$ be a measurable function.
Assume that the transition kernel of  $\{Y_t\}_{t=0}^\infty$ is at least polynomially ergodic of order  $q\geq (1+\gamma)(1+2/\gamma^*)$, and that $\pi^* (\|f\|^{2+\gamma^*})<\infty$ for some $\gamma>0$ and $\gamma^*>0$.
Define $\cov'(l) = E^* \{(f(Y_0) - \theta)(f(Y_{l}) -\theta)^\top\}$ if $l\geq 0$, and 
$\cov'(l) = E^* \{(f(Y_{-l}) -\theta) (f(Y_0) - \theta)^\top\}$ if $l<0$,
where $E^*(\cdot)$ denotes the expectation taken if $Y_0 \sim \pi^*$, and
define 
\begin{equation}\label{eqn:varsigmahomo}
	\Sigma^* = \sum_{l=-\infty}^{+\infty} \cov' (l).
\end{equation}
Assume further that $\Sigma^*$ is positive definite.
Then
the multivariate CLT holds:
\begin{equation}\label{eqn:homoclt}
	\sqrt m \left(\frac{\sum_{t=1}^m f(Y_t)}{m} - \pi^*(f)\right) \xrightarrow{d} \mathcal{N}(0, \Sigma^*), \text{ as } m \to \infty.
\end{equation}
\end{lemma}

\begin{proof}
	By \pcite{jones2004markov} Theorem 9 (2), a univariate CLT holds. Then the multivariate CLT~\eqref{eqn:homoclt} holds by the Cramér–Wold Theorem.
\end{proof}
	

We now prove Theorem \ref{thm:clt}. 
We shall begin with a few technical results.

By Lemma \ref{lem:1}, the transition kernel of the block chain is at least polynomially ergodic of order $q$. If $X_0\sim\pi$, by \pcite{jones2004markov} Section 3, the block chain is $\alpha$-mixing with $\alpha(n)\leq O(n^{q})$, where $q\geq (1+\gamma)(1+2/\gamma^*)$.

Recall that $\cov(j,l)$, $j\in\{0,\dots,k-1\}$ and $l\in\mathbb{Z}$, is the $(j,l)$th autocovariance matrix defined in \eqref{eqn:xi}.
We will show the absolute convergence of $\sum_{l=-\infty}^{+\infty}\sum_{j=0}^{k-1} |\cov_{a,b}(j, l)|,$ where $\cov_{a,b}(j, l)$ denotes the $(a,b)$th element of $\cov(j, l)$, and $a$ and $b$ are in $\{1,\dots,d\}$.
Since $\pi(\|f\|^{2+\gamma^*})<\infty$, for $X\sim\pi$, we have $\|f(X)\|_{2+\gamma^*}<\infty$, where $\|\cdot\|_p$ denote the $L_p$-norm of random vectors.
By~\pcite{kuelbs1980almost} Lemma 2.1 (see also \pcite{deo1973note} Lemma 1) and the $\alpha$-mixing coefficient of the block chain, we have  for all $a\in\{1,\dots,d\}$, $b\in\{1,\dots,d\}$, $j\in\{0,\dots,k-1\}$, and $l$ large enough,
$$|\cov_{a,b}(j, l)| \leq O(\lfloor l/k\rfloor^{-q})^{\gamma^*/(2+\gamma^*)}= O(l^{-q\gamma^*/(2+\gamma^*)})\leq O(l^{-1-\gamma}).$$
Therefore,
$\sum_{l=-\infty}^{+\infty}\sum_{j=0}^{k-1} |\cov_{a,b}(j, l)|<\infty.$

Consider the function $f': \mathcal{X}^k \to \mathbb{R}^d$, where $f'(x_1,\dots ,x_{k}) = f(x_1)+\cdots +f(x_{k})$.
Let $$
\Sigma' = \sum_{l=-\infty}^{0} E_{\pi} \{(f'(Z^{\scriptsize\mbox{blo}}_{-l} )   -   k\theta) (f'(Z^{\scriptsize\mbox{blo}}_0) - k\theta)^\top\} + \sum_{l=1}^{+\infty} E_{\pi} \{(f'(Z^{\scriptsize\mbox{blo}}_0) - k\theta)(f'(Z^{\scriptsize\mbox{blo}}_l) - k\theta)^\top\}.
$$
Define $\sum_{l=0}^{-1}\cov(j,l) = 0$.
By the series rearrangement theorem \citep[Theorem 3.55]{rudin1964principles}, we get
	$$
\begin{aligned}
	\Sigma'   
	=& \sum_{l=-\infty}^{-1} E_{\pi} \left\{\sum_{i=0}^{k-1} (f(X_{-lk+i} ) - \theta)\right\}\left\{\sum_{j=0}^{k-1} (f(X_{j} ) - \theta)^\top\right\} \\ 
	&+  E_{\pi} \left\{\sum_{i=0}^{k-1} (f(X_{i} ) - \theta)\right\}\left\{\sum_{j=0}^{k-1} (f(X_{j} ) - \theta)^\top\right\} \\
	& + \sum_{l=1}^{+\infty} E_{\pi}
	\left\{\sum_{j=0}^{k-1} (f(X_{j} ) - \theta)\right\}
	\left\{\sum_{i=0}^{k-1} (f(X_{lk+i} ) - \theta)^\top\right\}\\
	  =& \sum_{j=0}^{k-1} \left\{\sum_{l=-\infty}^{-1} \sum_{i=0}^{k-1} \cov(j, j+lk-i) + \sum_{l=j-k+1}^{-1}\cov(j,l)  + \sum_{l=0}^{k-1-j}\cov(j,l) +
	\sum_{l=1}^{+\infty} \sum_{i=0}^{k-1} \cov(j, lk+i-j)\right\}\\
	  =& \sum_{j=0}^{k-1} \left\{\sum_{l=-\infty}^{j-k}  \cov(j, l) + \sum_{l=j-k+1}^{-1}\cov(j,l)  + \sum_{l=0}^{k-j-1}\cov(j,l) +
	\sum_{l=k-j}^{+\infty}  \cov(j, l)\right\}\\
	 =& \sum_{l=-\infty}^{+\infty}\sum_{j=0}^{k-1} \cov(j, l).
\end{aligned}
$$
That is, $\Sigma' = k \Sigma$ where $\Sigma$ is given by \eqref{eqn:varsigma}.

For an arbitrary initial distribution, we apply Lemma \ref{lem:homoclt} to the block chain and function $f'$. 
Given the assumption in Theorem \ref{thm:clt} that $\Sigma$ is positive definite,  $\Sigma'$ is also positive definite. By Lemma \ref{lem:homoclt}, we have
\begin{equation} \nonumber
	\sqrt {m} \left(\sum_{t=0}^{mk-1} \frac{f(X_t)}{m} - k \theta \right) \xrightarrow{d} \mathcal{N}(0, \Sigma'), \text{ as } m \to \infty,
\end{equation}
It follows that
\begin{equation} \label{eqn:CLT-block-2}
	\sqrt {mk} \left(\sum_{t=0}^{mk-1} \frac{f(X_t)}{mk} - \theta \right) \xrightarrow{d} \mathcal{N}(0, \Sigma), \text{ as } m \to \infty,
\end{equation}

For a distribution $\nu$ on $(\mathcal{X}, \mathcal{F})$ and $A \in \mathcal{F}^{\mathbb{N}}$ (where $\mathcal{F}^{\mathbb{N}} = \prod_{i=0}^{\infty} \mathcal{F}$ means the product $\sigma$ algebra generated by cylinder sets), we use $P_{\nu}(\{X_t\}_{t=0}^{\infty} \in A)$ to denote the probability of $\{X'_t\}_{t=0}^{\infty} \in A$ where $\{X_t'\}$ is a chain with the same transition law as $\{X_t\}$ and $X'_0 \sim \nu$.
By the ergodicity of the transition kernel of the subchain $\{X_{kt+i-1}\}_{t=0}^{\infty}$, for $i \in \{1,\dots,k\}$ and $\varepsilon > 0$, 
\[
\begin{aligned}
	&P_{\nu}(\|f(X_{mk + i-1})\|/\sqrt{mk} \geq \varepsilon) \\
	\leq & P_{\pi}(\|f(X_{i-1})\|/\sqrt{mk} \geq \varepsilon) + \|\nu (K_i K_{i+1} \cdots K_k K_1 \cdots K_{i-1})^m(\cdot) -\pi(\cdot)\|_{\text{TV}} \to 0,
\end{aligned}
\]
as $m \to \infty$.
Then, regardless of the initial distribution, for $r \in \{1,\dots,k\}$,
\begin{equation} \label{eqn:CLT-block-3}
	 \sum_{t=mk}^{mk+r} \frac{f(X_t)}{\sqrt{mk}} \xrightarrow{P} 0, \text{ as } m \to \infty,
\end{equation}
By \eqref{eqn:CLT-block-2}, \eqref{eqn:CLT-block-3}, and Slutsky's theorem, for $r \in \{1,\dots,k\}$, the law of $\sqrt{mk+r}(\hat{\theta}_{mk+r-1} - \theta)$ converges weakly to $\mathcal{N}(0, \Sigma)$ as $m \to \infty$.
Evidently,  this implies that the law of $\sqrt{n}(\hat{\theta}_n - \theta)$ converges weakly to $\mathcal{N}(0, \Sigma)$ as $n \to \infty$.

\section{Proofs for Section \ref{ssec:sip}}

\subsection{Proof of Lemma \ref{lem:ergodic}} \label{proof:ergodic}

Suppose that $K_U$ is Harris ergodic.
Let $\nu$ be the starting distribution of $\{X_t\}_{t=0}^{\infty}$.
Let $\tilde{\pi}_U: \mathcal{F}_U^2 \to [0,1]$ be such that
\[
\tilde{\pi}_U(\df u, \df u') = \pi_U(\df u) K_U(u, \df u').
\]
Since $K_U$ is Harris ergodic, routine calculations show that the law of $(U_t, U_{t+1})$ converges in the total variation distance to $\tilde{\pi}_U$ as $t \to \infty$.
Recall that $X_{kt+k} = g_k(U_t, U_{t+1})$ for $t \in \mathbb{N}$.
Then the law of $X_{kt+k}$ converges to the distribution $\tilde{\pi}_U \circ g_k^{-1}$.
Since $\nu$ is arbitrary, $\tilde{K}_1$, the transition kernel of the first homogeneous subchain $\{Z_{t}^1\}_{t=0}^{\infty} = \{X_{kt}\}_{t=0}^{\infty}$, is Harris ergodic.

Suppose further that $K_U$ is at least polynomially ergodic of order $q$.
That is, for $u \in \mathcal{X}_U$ and $t \in \mathbb{N}_+$,
\[
\|K_U^t(u, \cdot)  - \pi_U(\cdot) \|_{\text{TV}} \leq \rho_U(u) t^{-q},
\]
where $\pi_U(\rho_U) < \infty$.
Let $x \in \mathcal{X}$.
Consider a copy of $\{U_t\}$ and $\{X_t\}$ such that $X_0 = x$, $U_t = g_0(X_{kt}, X_{kt+1}, \dots, X_{kt+k_0-1})$ for $t \in \mathbb{N}$, and $X_{kt+k} = g_k(U_t, U_{t+1})$ for $t \in \mathbb{N}_+$.
Denote the distribution of $U_0$ by $\nu$. 
Recall that $\tilde{K}_1 = K_1 \cdots K_k$.
Then, by the polynomial ergodicity of $K_U$, for $t \in \mathbb{N}_+$ and $A \in \mathcal{F}$,
\begin{equation} \label{ine:Kpoly}
\begin{aligned}
	| \tilde{K}_1^t(x, A) - \pi(A) | &= \left| \int_{\mathcal{X}_U^2} \{ \nu K_U^t (\df u) - \pi_U(\df u) \} K_U(u, \df u') \, \mathbb{I}\{ g_k(u,u') \in A \} \right| \\
	&\leq \|\nu K_U^t(\cdot) - \pi_U(\cdot)\|_{\text{TV}} \\
	&= \sup_{B \in \mathcal{F}_U} \left| \int_{\mathcal{X}_U} \nu(\df u) \left\{ K_U^t(u, B) - \pi_U(B) \right\} \right| \\
	&\leq \sup_{B \in \mathcal{F}_U}  \int_{\mathcal{X}_U} \nu(\df u) \left| K_U^t(u, B) - \pi_U(B) \right| \\
	&\leq \nu(\rho_U) t^{-q}.
\end{aligned}
\end{equation}
Note that, if $k_0 - 1 = kt'+r$ for some $t' \in \mathbb{N}$ and $r \in \{1,\dots,k\}$,
\[
\nu(\rho_U) = \rho(x) := \int_{\mathcal{X}^{k_0-1}} \rho_U\{ g_0(x, x_1, \dots, x_{k_0-1}) \} K_1(x, \df x_1) \, \cdots \, K_r(x_{k_0-2}, \df x_{k_0-1}),
\]
so
\[
\pi(\rho) = \int_{\mathcal{X}^{k_0}} \rho_U\{ g_0(x, x_1, \dots, x_{k_0-1}) \} \, \pi(\df x) \, K_1(x, \df x_1) \, \cdots \, K_r(x_{k_0-2}, \df x_{k_0-1}).
\]
(If $k_0 = 1$ then $\pi(\rho)$ is just $\pi(\rho_U \circ g_0)$.)
It was already shown that $\tilde{K}_1$ is Harris ergodic, so as $t \to \infty$, the law of $X_{kt}$ converges in the total variation distance to $\pi$.
Then, by the cyclic nature of the transition laws of $\{X_t\}_{t=0}^{\infty}$, the law of $U_t = g_0(X_{kt}, X_{kt+1}, \dots, X_{kt+k_0-1})$ must converge in the total variation distance to the law of $g_0(X_0', X_1', \dots, X_{k_0-1}')$ where $X_0' \sim \pi$ and $\{X_t'\}_{t=0}^{\infty}$ is a Markov chain with the same transition laws as $\{X_t\}_{t=0}^{\infty}$.
The limiting distribution of $\nu K_U^t$ must coincide with the stationary distribution of $K_U$, so the law of $g_0(X_0', X_1', \dots, X_{k_0-1}')$ is $\pi_U$.
The most recent display then implies that $\pi(\rho) = \pi_U(\rho_U) < \infty$.
It then follows from \eqref{ine:Kpoly} that the first homogeneous subchain $\tilde{K}_1$ is at least polynomially ergodic of order $q$.

The proof for geometric ergodicity is analogous.

\subsection{Proof of Lemma \ref{lem:tau}}\label{proof:tau}

	Denote by $\check{E}_{\nu}$ the expectation of a function of the split chain if $(U_0, \delta_0) \sim \nu$, where $\nu$ is an arbitrary initial distribution.
	To be more precise, for a function $V$, $\check{E}_{\nu} V(\{U_t, \delta_t\}_{t=0}^{\infty})$ denotes the expected value of $V(\{U_t,\delta_t\}_{t = 0}^{\infty})$ when the chain $\{U_t,\delta_t\}_{t = 0}^{\infty}$ is redefined on some probability space so that it has the same transition law as before, but $(U_0,\delta_0) \sim \nu$.
	Our goal is to show that $\check{E}_{\nu} (T_2 - T_1)^{q'} < \infty$ where $q'$ is an arbitrary number in $[0,q+1)$.
	
	Let $S = \inf\{t > 0: \ \delta_t = 1\}$.
	Note that $T_2 - T_1 = \inf\{t > 0: \ \delta_{T_1 - 1 + t} = 1\}$, $T_1 - 1$ is a stopping time, and $\delta_{T_1-1} = 1$.
	Then, by the strong Markov property~\citep[see, e.g.,][Proposition 3.4.6]{meyn2005markov} and the split chain construction, $\check{E}_{\nu} (T_2 - T_1)^{q'} = \check{E}_{\varpi} S^{q'}$, where $\varpi$ is any distribution concentrated on $\mathcal{X}_U \times \{1\}$.
	Indeed, $\mathcal{X}_U \times \{1\}$ is an atom for the split chain, so $T_2 - T_1 = (T_2 - 1) - (T_1 - 1)$, being the time difference between the first and second times that chain is in the atom, has the same distribution as $S$, which is the first return time to the atom if the split chain starts from the atom.
	Hence, it suffices to show that $\check{E}_{\varpi} S^{q'} < \infty$.

	By a de-initialization argument \citep[][Theorem~1]{roberts2001markov}, the transition kernel of $\{U_t, \delta_t\}_{t=0}^{\infty}$ is at least polynomially ergodic of order $q$.
	Then, if $(U_0, \delta_0) \sim \pi_U^*$, where $\pi_U^*$ is the stationary distribution of the split chain $\{U_t, \delta_t\}_{t=0}^{\infty}$, the process $\{U_t, \delta_t\}_{t=0}^{\infty}$ is  $\alpha$-mixing with $\alpha(n)\leq Cn^{-q}$, where $C$ is a constant \citep[Section 3]{jones2004markov}. 
	Let $q^*=\max\{2,q'\}$.
	Then, recalling that $q > 1$  and $q' < q+1$, we see that  $\sum_{n=1}^{\infty} n^{q^*-2} n^{-q} < \infty$, and $q^*-2\geq 0$.
	It then follows from \pcite{bolthausen1982berry} Lemma 3 that $\check{E}_{\varpi} S^{q^*} < \infty$. Thus, we have $\check{E}_{\varpi} S^{q'} < \infty$.

\subsection{Proof of Lemma \ref{lem:EY}}\label{proof:EY}

By Lemma \ref{lem:finite}, $T_1$ is finite almost surely.
Since $\{\tau_t\}_{t=1}^\infty$ are i.i.d. random variables, and $T_{2n+1} = \sum_{t=1}^{2n}\tau_t + T_1$, 
by the classic SLLN, we have with probability 1, 
\begin{equation}\label{eqn:sllnT}
	\frac{1}{2n}T_{2n+1} = \frac{1}{2n}\left(\sum_{t=1}^{2n} \tau_t \right) + \frac{1}{2n}T_1 \to  \Xi_{\tau}, \text{ as } n\to\infty.
\end{equation}
Since with probability 1, $T_{2n+1}\to\infty$ as $n\to\infty$, by Lemma~\ref{lem:ergodic} and Theorem \ref{thm:slln}, we have with probability 1,
\begin{equation}\label{eqn:sllnx}
	\frac{1}{kT_{2n+1}}\left(\sum_{t=1}^{kT_{2n+1}}f(X_t) \right)\to\theta, \text{ as } n\to\infty.
\end{equation}
Noticing that 
$$
\sum_{t=1}^{2n} Y_t= \sum_{t=1}^{2n}\sum_{i=kT_{t}+1}^{kT_{t+1}} f(X_i) = \sum_{t=kT_1+1}^{kT_{2n+1}}f(X_t) = \sum_{t=1}^{kT_{2n+1}}f(X_t) - \sum_{t=1}^{kT_1}f(X_t),
$$
by \eqref{eqn:sllnT} and \eqref{eqn:sllnx}, we have with probability 1,
$$
\frac{1}{2n}\left(\sum_{t=1}^{2n} Y_i\right)  = \frac{T_{2n+1}}{2n}\frac{1}{T_{2n+1}}\left(\sum_{t=1}^{kT_{2n+1}}f(X_t) \right) -  \frac{1}{2n}\left(\sum_{t=1}^{kT_1}f(X_t)\right)\to \Xi_{\tau} k\theta, \text{ as } n\to\infty.
$$

On the other hand, it is easy to see that $\{Y_{2t-1}\}_{t=1}^\infty$ are i.i.d. random vectors, and so are $\{Y_{2t}\}_{t=1}^\infty$.
Then, by the classic SLLN, we have with probability 1,
\begin{equation}\nonumber
	\frac{1}{2n}\left(\sum_{t=1}^{2n} Y_i\right) = \frac{1}{2n}\left(\sum_{t=1}^n Y_{2t-1}\right)+\frac{1}{2n}\left(\sum_{t=1}^n Y_{2t}\right) \to  \Xi_Y, \text{ as } n\to\infty.
\end{equation}
Combining the two most recent displays shows that $\Xi_Y = k\Xi_{\tau}\theta<\infty$.

\subsection{Proof of Lemma \ref{lem:4}}\label{proof:lem4}

Denote by $\check{E}_{\nu}$ the expectation of a function of $\{U_t, \delta_t\}_{t=0}^{\infty}$ if $(U_0, \delta_0) \sim \nu$.
Denote by $\hat{E}_{\nu'}$ the expectation of a function of $\{U_t\}_{t=0}^{\infty}$ if $U_0 \sim \nu'$.
Again, to be precise, when we write $\hat{E}_{\nu'} V(\{U_t\}_{t=0}^{\infty})$ for some function~$V$, we are redefining the chain $\{U_t\}_{t=0}^{\infty}$ on some probability space so that it has the same transition law as before, but $U_0 \sim \nu'$.
For $t \in \mathbb{N}$ and $i \in \{1,\dots,k\}$, when $X_{kt + i}$ appears in an expectation like $\hat{E}_{\nu'}$, it is redefined in terms of the redefined sequence of $\{U_n\}_{n=0}^{\infty}$ through the formula $X_{tk+i} = g_i(U_t, U_{t+1})$.
The same goes when $X_{kt + i}$ appears in an expectation like $\check{E}_{\nu}$.

Let $\varpi$ be some distribution concentrated on $\mathcal{X}_U \times \{1\}$.
By Condition \ref{con:min}, we have  for all $A\in\mathcal{F}_U$,
\begin{equation}\nonumber
	\pi_U(A) = \pi_U K_U(A) = \int_{\mathcal{X}_U} \pi_U(du)K(u, A)\geq \mu(A)\int_{\mathcal{X}_U} h(u) \,\pi_U(du) = \mu(A) \pi_U(h).
\end{equation}
Then, by the Markov property, for $t\in\mathbb{N}_+$ and $i \in \{1,\dots.k\}$,
\begin{equation} \label{ine:Efmoment-1}
	\begin{aligned}
		& \check{E}_{\varpi} (\|f(X_{kt+i})\|^{2+\gamma+\gamma^*})  = \int_{\mathcal{X}_U} \check{E}_{\varpi} \{\|f(X_{kt+i})\|^{2+\gamma+\gamma^*} \mid U_1=u\} \, \mu(\df u)  \\ 
		& \leq  \frac{1}{\pi_U(h)}  \int_{\mathcal{X}_U} \check{E}_{\varpi} \{\|g_i(U_t, U_{t+1})\|^{2+\gamma+\gamma^*} \mid U_1=u\} \, \pi_U(\df u)  \\ 
		&= \frac{1}{\pi_U(h)}  \hat{E}_{\pi_U} \|g_i(U_{t-1}, U_t)\|^{2+\gamma+\gamma^*} \\
		&=  \frac{1}{\pi_U(h)}  \hat{E}_{\pi_U} \|g_i(U_0, U_1)\|^{2+\gamma+\gamma^*}.
	\end{aligned}
\end{equation}
Define the distribution $\pi_i: \mathcal{F} \to [0,1]$ so that 
\[
\pi_i(A) = \hat{E}_{\pi_U} \mathbb{I}\{ g_i(U_0, U_1) \in A \} = \int_{\mathcal{X}_U^2} \mathbb{I}\{ g_i(u_0, u_1) \in A \} \, \pi_U(\df u_0) \, K_U(u_0, \df u_1).
\] 
Since $K_U$ is Harris ergodic, for the original sequence of $\{X_t\}_{t=0}^{\infty}$, the distribution of $X_{kt+i} = g_i(U_t, U_{t+1})$ converges in the total variation distance to $\pi_i$ as $t \to \infty$.
By Lemmas~\ref{lem:1} and~\ref{lem:ergodic}, the transition kernel of the subchain $\{X_{kt+i}\}_{t=0}^{\infty}$ is Harris ergodic, so $\pi$ is the unique limiting distribution of $\{X_{kt+i}\}_{t=0}^{\infty}$, and $\pi = \pi_i$.
It then follows from \eqref{ine:Efmoment-1} that, for $t \geq k+1$,
\begin{equation} \label{ine:Efmoment-2}
	\check{E}_{\varpi} \| f(X_t) \|^{2+\gamma+\gamma^*}  \leq \frac{1}{\pi_U(h)}  \pi(\|f\|^{2+\gamma+\gamma^*}) < \infty.
\end{equation}

To continue, note that
$$
\check{E}_{\varpi} \left\| Y_1\right\|^{2+\gamma}
\leq \check{E}_{\varpi} \left(\sum_{t=kT_1 + 1}^{kT_{2}} \| f(X_t)  \|\right)^{2+\gamma} \leq \check{E}_{\varpi} \left(\sum_{t=k+1}^{\infty}\mathbb{I}(t\leq kT_{2})  \| f(X_t)  \|\right)^{2+\gamma}.
$$
By  Minkowski's inequality,
$$
\left\{\check{E}_{\varpi} \left(\sum_{t=k+1}^{\infty}\mathbb{I}(t\leq kT_{2})  \| f(X_t)  \|\right)^{2+\gamma}\right\}^{1/(2+\gamma)}
\leq \sum_{t=k+1}^{\infty} \left\{ \check{E}_{\varpi} \, \mathbb{I}(t\leq kT_{2}) \| f(X_t)  \|^{2+\gamma} \right\}^{1/(2+\gamma)}.
$$
By Hölder's inequality, 
$$
\begin{aligned}
	& \check{E}_{\varpi} \, \mathbb{I}(t\leq kT_{2}) \| f(X_t)  \|^{2+\gamma} \\
	& \leq 
	\left\{\check{E}_{\varpi} \mathbb{I}(t\leq kT_{2}) \right\}^{\gamma^*/(2+\gamma+\gamma^*)} \left( \check{E}_{\varpi} \| f(X_t) \|^{2+\gamma+\gamma^*}  \right) ^{(2+\gamma)/(2+\gamma+\gamma^*)}.  
\end{aligned}
$$	
Therefore, 
\begin{equation}\label{eqn:y20}	
	\begin{aligned}
		& \left(\check{E}_{\varpi} \left\| Y_1\right\|^{2+\gamma}\right)^{1/(2+\gamma)} 
		\\ 
		&\leq \left\{ \check{E}_{\varpi} \left(\sum_{t=kT_1 + 1}^{kT_{2}} \| f(X_t)  \|\right)^{2+\gamma} \right\}^{1/(2+\gamma)} \\
		& \leq \sum_{t=k+1}^{\infty}	\left\{	\left\{ \check{E}_{\varpi} \mathbb{I}(t\leq kT_{2}) \right\}^{\gamma^*/(2+\gamma+\gamma^*)}\left( \check{E}_{\varpi} \| f(X_t)^{2+\gamma+\gamma^*}  \|\right) ^{(2+\gamma)/(2+\gamma+\gamma^*)}\right\}^{1/(2+\gamma)}     \\
		& \leq \left\{ \frac{1}{\pi_U(h)}\pi(\| f\|^{2+\gamma+\gamma^*}) \right\}^{1/(2+\gamma+\gamma^*)} \sum_{t=k+1}^{\infty}		\left\{\check{E}_{\varpi} \mathbb{I}(t\leq kT_{2}) \right\}^{\gamma^*/\{(2+\gamma+\gamma^*)(2+\gamma)\}} ,
	\end{aligned}
\end{equation}
where the last line follows from \eqref{ine:Efmoment-2}.
Let $q'$ be a constant such that $(2+\gamma)\{1+(2+\gamma)/\gamma^*\}<q'<q+1$. 
Let $S = \inf\{t > 0: \, \delta_t = 1\}$.
Then $T_1 \leq S+1$.
It was shown in the proof of Lemma~\ref{lem:tau} that $\check{E}_{\varpi} S^{q'} = \check{E}_{\varpi} \tau_1^{q'} < \infty$.
By Jensen's inequality,
\[
\check{E}_{\varpi} T_2^{q'} = \check{E}_{\varpi} (T_1 + \tau_1)^{q'} \leq 2^{q'-1} ( \check{E}_{\varpi} T_1^{q'} + \check{E}_{\varpi} \tau_1^{q'} ) < \infty.
\]
By  Markov's inequality and the fact that $q'\gamma^*/\{(2+\gamma+\gamma^*)(2+\gamma)\}>1$,
$$
\begin{aligned}
	& \sum_{t=k+1}^{\infty}	\left\{\check{E}_{\varpi} \mathbb{I}(t\leq kT_{2}) \right\}^{\gamma^*/\{(2+\gamma+\gamma^*)(2+\gamma)\}} \\
	 =& \sum_{t=k+1}^{\infty}	\left(	\check{E}_{\varpi} \mathbb{I} \left\{T_2^{q'}\geq (t/k)^{q'}\right\} \right)^{\gamma^*/\{(2+\gamma+\gamma^*)(2+\gamma)\}} \\
	\leq &  \left\{ \check{E}_{\varpi}  T_2^{q'}  \right\}^{\gamma^*/\{(2+\gamma+\gamma^*)(2+\gamma)\}}\sum_{t=k+1}^{\infty} \left(t/k\right)^{-q'\gamma^*/\{(2+\gamma+\gamma^*)(2+\gamma)\} }<\infty.
\end{aligned}
$$
Therefore, by \eqref{eqn:y20},
\[
\check{E}_{\varpi} \left\| Y_1\right\|^{2+\gamma}
\leq \check{E}_{\varpi} \left(\sum_{t=kT_1 + 1}^{kT_{2}} \| f(X_t)  \|\right)^{2+\gamma} < \infty.
\]
By Lemma \ref{lem:tau}, $\check{E}_{\varpi} {\tau_1}^{2+\gamma} <\infty$. 
Thus, 	$\check{E}_{\varpi} \left\|\tilde{Y}_1\right\|^{2+\gamma}<\infty.$
But the distributions of the $\Delta_i$'s do not depend on the initial distribution of the split chain \citep[][Section 5.3]{nummelin2004general}, so these expectations are finite regardless of the initial distribution.

\subsection{Proof of Lemma \ref{lem:vfvar}}\label{proof:vfvar}

We prove Lemma \ref{lem:vfvar} by contradiction. 
To this end assume that that $\Sigma_Y$ is not positive definite.
Then there exists $a\in\mathbb{R}^d$ such that $a \neq 0$ and $a^\top \Sigma_Ya =0$. 
We have
$$
a^\top\Sigma_Y a = \lim_{n\to\infty} n^{-1} \mathrm{Var}\left(\sum_{i=1}^n a^\top\tilde{Y}_i\right) = \mathrm{Var}(a^\top\tilde{Y}_1) + 2\mathrm{Cov}(a^\top\tilde{Y}_1, a^\top\tilde{Y}_{2}) = 0.
$$
For all $m\geq 2$, 
$$
\mathrm{Var}\left(\sum_{i=1}^m a^\top\tilde{Y}_i\right) = m\mathrm{Var}(a^\top\tilde{Y}_1) + 2(m-1)\mathrm{Cov}(a^\top\tilde{Y}_1, a^\top\tilde{Y}_{2}) = -2\mathrm{Cov}(a^\top\tilde{Y}_1, a^\top\tilde{Y}_{2})
$$
is a constant.

Let
$	\xi(n) = \sup\{ i: T_i \leq n\}.$ 
Let $\lceil\cdot\rceil$ and $\lfloor\cdot\rfloor$ denote the ceiling and floor functions, respectively. 
Define $n^* = \lceil n/\Xi_{\tau} \rceil$. 
By Chebyshev's inequality, for all $\epsilon>0$, we have 
$$
	\begin{aligned}
		P\left(\left|n^{-1/2}\sum_{i=1}^{n^*}a^\top\tilde{Y}_i\right|\geq\epsilon\right)\leq 
		n^{-1}\epsilon^{-2}\mathrm{Var}\left(\sum_{i=1}^{n^*} a^\top\tilde{Y}_i\right) \\ = -2 n^{-1}\epsilon^{-2}\mathrm{Cov}(a^\top\tilde{Y}_1, a^\top\tilde{Y}_{2})\to 0,\text{ as }n\to\infty.
	\end{aligned}
$$
Therefore, 
\begin{equation}\label{eqn:xis}
n^{-1/2} \sum_{i=1}^{n^*}a^\top\tilde{Y}_i \xrightarrow{P} 0, \text{ as }n\to\infty.
\end{equation}

We now use an argument from Section 17.2.2 of \cite{meyn2005markov} to show that~\eqref{eqn:xis} continues to hold when $n^*$ is replaced by a particular random time.
Let $\epsilon>0$ be arbitrary.
Let $\underline n = \lceil (1/\Xi_{\tau}-\epsilon) n\rceil$, $\overline n = \lfloor (1/\Xi_{\tau}+\epsilon) n\rfloor$.
By (iv) of Lemma \ref{lem:error}, with probability 1, $\xi(n)/n \to 1/\Xi_{\tau}$ as $n\to\infty$.
Then, there exists some $n_{\epsilon}\in\mathbb{N}_+$ such that, for all $n>n_{\epsilon}$,
$$
P(\underline n \leq \xi(n)-1 \leq \overline n)\geq 1-\epsilon.
$$
For $n>n_{\epsilon}$ and $b>0$,
\begin{equation} \label{eqn:pmstart}
\begin{aligned}
& P\left( \left| n^{-1/2} \sum_{i=1}^{\xi(n)-1}a^\top\tilde{Y}_i-  n^{-1/2} \sum_{i=1}^{n^*}a^\top\tilde{Y}_i \right|>b\right)\\
& \leq \epsilon 
+ P\left( \max_{\underline n\leq t\leq n^*} \left|\sum_{i=t}^{n^*}a^\top\tilde{Y}_i\right|>b n^{1/2} \right)
+P\left( \max_{n^* \leq t\leq \overline n} \left|\sum_{i=n^*}^{t}a^\top\tilde{Y}_i\right|>b n^{1/2} \right).
\end{aligned}
\end{equation}
Notice that
\begin{equation}\label{eqn:pm0}
\begin{aligned}
	& P\left( \max_{\underline n\leq t\leq n^*} \left|\sum_{i=t}^{n^*}a^\top\tilde{Y}_i\right|>b n^{1/2} \right)\\
	& \leq
	P\left( \max_{\underline n\leq t\leq n^*} \left|\sum_{i\in I_t^{\text{even}}}a^\top\tilde{Y}_i\right|>b n^{1/2}/2 \right)+
	P\left( \max_{\underline n\leq t\leq n^*} \left|\sum_{i\in I_t^{\text{odd}}}a^\top\tilde{Y}_i\right|>b n^{1/2}/2 \right),
\end{aligned}
\end{equation}
where $I_t^{\text{even}} = \{2j: t \leq 2j \leq n^*, j\in\mathbb{N}\}$, and $I_t^{\text{odd}} = \{2j+1: t \leq 2j+1 \leq n^*, j\in\mathbb{N}\}$.
Since $\{a^\top\tilde{Y}_t\}_{t=1}^\infty$ is a mean-zero stationary 1-dependent sequence, by Kolmogorov's inequality \citep[see, e.g.,][Theorem D.6.3]{meyn2005markov}, we have
$$
P\left( \max_{\underline n\leq t\leq n^*} \left|\sum_{i\in I_t^{\text{even}}}a^\top\tilde{Y}_i\right|>b n^{1/2}/2 \right)\leq \frac{n\epsilon/2+2}{nb^2/4}\mathrm{Var}(a^\top \tilde{Y}_1).
$$
Similarly,
$$
P\left( \max_{\underline n\leq t\leq n^*} \left|\sum_{i\in I_t^{\text{odd}}}a^\top\tilde{Y}_i\right|>b n^{1/2}/2 \right)\leq \frac{n\epsilon/2+2}{nb^2/4}\mathrm{Var}(a^\top \tilde{Y}_1).
$$
Therefore,
\begin{equation}\label{eqn:pm1}
P\left( \max_{\underline n\leq t\leq n^*} \left|\sum_{i=t}^{n^*}a^\top\tilde{Y}_i\right|>b n^{1/2} \right)\leq \frac{4n\epsilon+16}{nb^2}\mathrm{Var}(a^\top \tilde{Y}_1).
\end{equation}
Using the same argument, we have
\begin{equation}\label{eqn:pm2}
	P\left(\max_{n^* \leq t\leq \overline n} \left|\sum_{i=n^*}^{t}a^\top\tilde{Y}_i\right|>b n^{1/2} \right)\leq \frac{4n\epsilon+12}{nb^2}\mathrm{Var}(a^\top \tilde{Y}_1).
\end{equation}
By \eqref{eqn:pmstart} to \eqref{eqn:pm2}, we have 
$$
 P\left( \left| n^{-1/2} \sum_{i=1}^{\xi(n)-1}a^\top\tilde{Y}_i-  n^{-1/2} \sum_{i=1}^{n^*}a^\top\tilde{Y}_i \right|>b\right)\leq \epsilon+  \frac{8n\epsilon+28}{nb^2}\mathrm{Var}(a^\top \tilde{Y}_1).
$$
Since $\epsilon>0$ is arbitrary, 
\begin{equation}\label{eqn:pm3}
  n^{-1/2} \sum_{i=1}^{\xi(n)-1}a^\top\tilde{Y}_i-  n^{-1/2} \sum_{i=1}^{n^*}a^\top\tilde{Y}_i  \xrightarrow{P}0,  \text{ as }n\to\infty.
\end{equation}
By \eqref{eqn:xis} and \eqref{eqn:pm3}, 
\begin{equation} \label{eqn:pmend}
n^{-1/2} \sum_{t=1}^{\xi(n)-1} a^\top \tilde{Y}_i \xrightarrow{P} 0,  \text{ as }n\to\infty.
\end{equation}

Since 
$$
\sum_{t=1}^{\xi(n)-1} a^\top \tilde{Y}_i = \sum_{t=kT_1+1}^{kT_{\xi(n)}} a^\top f(X_t) - a^\top (T_{\xi(n)} - T_1) k \theta,
$$
by (i)-(iii) of Lemma~\ref{lem:error} and \eqref{eqn:pmend},  for $r \in \{1,\dots,k\}$,
$$
n^{-1/2} \sum_{t=1}^{kn+r} \left(a^\top f(X_t) - a^\top k \theta \right)  \xrightarrow{P} 0, \text{ as }n\to\infty.
$$
On the other hand, since $\Sigma$ is positive definite and $a \neq 0$, $a^\top \Sigma a>0$. Applying Lemma~\ref{lem:ergodic} and Theorem \ref{thm:clt}, we have 
$$
n^{-1/2} \sum_{t=1}^{kn+r} (a^\top f(X_t) - a^\top \theta) \xrightarrow{d}\mathcal{N}(0, ka^\top \Sigma a), \text{ as }n\to\infty.
$$
The two most recent displays contradict with each other.
Therefore, $\Sigma_Y$ must be positive definite.

\subsection{Proof of Lemma \ref{lem:error}}\label{proof:lemerror}

By Lemma \ref{lem:finite}, $T_1$ is almost surely finite. 
So (i) holds.

Next, observe that with probability 1, $\xi(n) \to \infty$.
To see this, note that $T_{\xi(n)+1} = T_1 + \sum_{t=1}^{\xi(n)} \tau_t$, and, 
by Lemma \ref{lem:tau}, $\tau_t<\infty$ for $t \in \mathbb{N}_+$ almost surely. 
Since $n\leq T_{\xi(n)+1}$, it holds that $P(\lim_{n\to\infty}T_{\xi(n)+1} =\infty)=1$.
If $P(\liminf_{n\to\infty} \xi(n)<\infty) >0$, then  $P(\liminf_{n\to\infty}T_{\xi(n)+1} <\infty)>0$, which leads to a contradiction.

To prove (ii), let $\tilde{Q}_{n,r} = \sum_{t=kT_{\xi(n)}+1}^{kn+r}  f(X_t)$ for $n \in \mathbb{N}$ and $r \in \{1,\dots,k\}$, and let $Q_m = \sum_{t=kT_{m}+1}^{kT_{m+1}}  \|f(X_t)\|$ for $m\in\mathbb{N}_+$. 
By the regeneration property, 
$
E Q_m ^{2+\gamma}= E Q_1  ^{2+\gamma}$,
and by Lemma \ref{lem:4},  $E Q_1  ^{2+\gamma}<\infty$.
Then
$$
\begin{aligned}
	\sum_{m=1}^\infty P(Q_m > m^{1/(2+\gamma)}) = \sum_{m=1}^\infty P(Q_1^{2+\gamma} > m) 
	\leq \sum_{m=1}^\infty \int_{m-1}^{m} P(Q_1^{2+\gamma} > x)dx \\= \int_{0}^\infty P(Q^{2+\gamma}_1 > x)dx = E Q_1^{2+\gamma}<\infty,	
\end{aligned}
$$
where  the last equality holds by \pcite{cinlar2013introduction} Theorem 1.9. 
By the Borel-Cantelli Lemma, we have with probability 1,
$\limsup_{m\to\infty} m^{-1/(2+\gamma)}Q_m \leq 1.$  Therefore, with probability 1, 
\begin{equation}\label{eqn:r}
	Q_m = O(m^{1/(2+\gamma)}), \text{ as }m\to\infty.
\end{equation}
Since with probability 1, $\xi(n)\to\infty$ as $n\to\infty$, with probability 1,
$$
Q_{\xi(n)} = O\left((\xi(n))^{1/(2+\gamma)}\right) = O(n^{1/(2+\gamma)}), \text{ as } n\to\infty.
$$
By the triangle inequality,
$$ \|\tilde{Q}_{n,r}\|\leq \sum_{t=kT_{\xi(n)}+1}^{kn+r}  \|f(X_t)\| \leq \sum_{t=kT_{\xi(n)}+1}^{kT_{\xi(n)+1}}  \|f(X_t)\|=Q_{\xi(n)}.$$ 
Combing the two most recent displays yields (ii).

We now establish (iii).
By Lemma \ref{lem:tau} and regeneration property, for all $m\in\mathbb{N}_+$, $E \tau_{m}^{2+\gamma}<\infty$. 
In the same way \eqref{eqn:r} is derived, we get with probability 1,
$$
\tau_m = O(m^{1/(2+\gamma)}), \text{ as } m\to\infty.
$$
Therefore, with probability 1,
$$
\tau_{\xi(n)} = O\left((\xi(n))^{1/(2+\gamma)}\right) = O(n^{1/(2+\gamma)}), \text{ as } n\to\infty.
$$
Since $|n-T_{\xi(n)}|\leq \tau_{\xi(n)}$, (iii) holds.

It remains to establish (iv).
By Lemma \ref{lem:tau},
$E\tau_1^{q'}<\infty$.
By the Marcinkiewics–Zygmund strong law of large numbers \citep[Introduction]{boukhari2020marcinkiewics}, with probability 1,
\begin{equation}\nonumber
	\lim_{m\to\infty}m^{-1/q'}\sum_{t=1}^m (\tau_t - \Xi_{\tau}) =0.
\end{equation}
It follows that, with probability 1,
$$
\lim_{n\to\infty}(\xi(n))^{-1/q'}\left(T_{\xi(n)+1} - \xi(n)\Xi_{\tau}\right) =0, \text{ and }
\lim_{n\to\infty}(\xi(n))^{-1/q'}\left\{T_{\xi(n)} - (\xi(n)-1)\Xi_{\tau}\right\} =0.
$$
Since $ T_{\xi(n)}\leq n \leq T_{\xi(n)+1}$, we have 
$$ T_{\xi(n)} - (\xi(n)-1)\Xi_{\tau} -  \Xi_{\tau}   \leq  n-\xi(n)\Xi_{\tau}\leq  T_{\xi(n)+1}-\xi(n)\Xi_{\tau}.$$
Therefore, with probability 1,
$$
\lim_{n\to\infty} \xi(n)^{-1/q'}\left(n - \xi(n)\Xi_{\tau}\right) =0, \text{ and } \lim_{n\to\infty} n / \xi(n) = \Xi_{\tau} .
$$
We have with probability 1, 
$$
\lim_{n\to\infty} \left({\frac{n}{\xi(n)}}\right)^{-1/q'} \left\{{(\xi(n))}^{-1/q'}(n - \xi(n)\Xi_{\tau})\right\} = 0.
$$
Thus, with probability 1, 
\begin{equation} \nonumber
	\lim_{n\to\infty} {n}^{-1/q'}(n - \xi(n)\Xi_{\tau}) = 0.
\end{equation}

\subsection{Proof of Lemma \ref{lem:measureexist}}\label{proof:measure}

First, we review some well-known topological results. 
\begin{lemma}\label{lem:product}
	For $i\in\mathbb{N}_+$, let $\mathcal{X}_i$ be a Polish space.
	Let $\mathbb{J}$ be a subset of $\mathbb{N}_+$.
	Then each of the following holds: 
	\begin{enumerate}
		\item [(i)] The space $\mathcal{X} = \prod_{i \in \mathbb{J}} \mathcal{X}_i$ equipped with the product topology is Polish.
		\item [(ii)] The Borel algebra of $\mathcal{X}$ coincides with the $\sigma$ algebra generated by cylinder sets, i.e., sets of the form $\{ (x_i)_{i \in \mathbb{J}} \in \mathcal{X}: \, x_j \in A_j \}$ where $j \in \mathbb{J}$ and $A_j \in \mathcal{B}(\mathcal{X}_j)$.
	\end{enumerate}
\end{lemma}

\begin{proof} [Proof of Lemma \ref{lem:measureexist}]
	Let $\mathcal{X}_A = \prod_{i=1}^{\infty} \mathbb{R}^{kd+1}$, and let $\mathcal{X}_B = \mathcal{X}_C = \prod_{i=1}^{\infty} \mathbb{R}^{d}$, each equipped with the corresponding product topology.
	By (i) of Lemma \ref{lem:product}, $\mathcal{X}_A$, $\mathcal{X}_B$, and $\mathcal{X}_C$ are Polish spaces.
	Denote by $\mathcal{B}(\mathcal{X}_A)$, $\mathcal{B}(\mathcal{X}_B)$, and $\mathcal{B}(\mathcal{X}_C)$ their respective Borel algebras.
	
	Denote by $(\Omega_1, \mathcal{F}_1, P_1)$ the probability space on which $\{X_t\}_{t=0}^{\infty}$ and $\{\delta_t\}_{t=1}^{\infty}$ are defined, and recall that the sequence $\{\tilde{Y}_i\}_{i=1}^{\infty}$ is determined by  $\{f(X_{kt+1}),\dots, f(X_{kt+k}), \delta_{t}\}_{t=0}^\infty$ through~\eqref{eqn:Y}.
	Recall from the proof of Theorem~\ref{thm:sip} that, by \pcite{liu2009strong} Theorem 2.1,  on another probability space $(\Omega_2, \mathcal{F}_2, \mathcal{P}_2)$, one can redefine the sequence  $\{\tilde{Y}_i\}_{i=1}^\infty$ together with $\{C(t)\}_{t=1}^\infty$, where 
	$\{C(t)\}_{t=1}^\infty$ are $d$ dimensional standard normal random vectors, such that
	with probability 1, \eqref{eqn:sippre} holds with $B(m)=\sum_{t=1}^m C(t)$. 
	We may view $(\{f(X_{kt+1}),\dots, f(X_{kt+k}), \delta_{t}\}_{t=0}^\infty, \{\tilde{Y}_t\}_{t=1}^\infty)$ as a function from $(\Omega_1,\mathcal{F}_1)$ to $(\mathcal{X}_A \times \mathcal{X}_B, \mathcal{B}(\mathcal{X}_A) \times \mathcal{B}(\mathcal{X}_B))$.
	We may view $(\{\tilde{Y}_t\}_{t=1}^\infty, \{B(t)\}_{t=1}^\infty)$ as a function from $(\Omega_2,\mathcal{F}_2)$ to $(\mathcal{X}_B \times \mathcal{X}_C, \mathcal{B}(\mathcal{X}_B) \times \mathcal{B}(\mathcal{X}_C))$.
	By (ii) of Lemma, $\mathcal{B}(\mathcal{X}_A) \times \mathcal{B}(\mathcal{X}_B)$ and $\mathcal{B}(\mathcal{X}_B) \times \mathcal{B}(\mathcal{X}_C)$ are large enough to host the events \eqref{eqn:Y} and \eqref{eqn:sippre}, and small enough so that the two functions above are measurable.
	Let $P_{AB}$ and $P_{BC}$ be the respective laws of these two functions, which are distributions on $(\mathcal{X}_A \times \mathcal{X}_B, \mathcal{B}(\mathcal{X}_A) \times \mathcal{B}(\mathcal{X}_B))$ and $(\mathcal{X}_B \times \mathcal{X}_C, \mathcal{B}(\mathcal{X}_B) \times \mathcal{B}(\mathcal{X}_C))$, respectively.
	
	By the gluing lemma on Polish spaces \citep[see, e.g.,][Lemma 3.1]{Ambrosio2013}, there exists a probability measure $P_{ABC}$ on $(\mathcal{X}_{ABC}, \mathcal{B}(\mathcal{X}_{ABC})) := (\mathcal{X}_A\times\mathcal{X}_B\times\mathcal{X}_C, \mathcal{B}(\mathcal{X}_A)\times\mathcal{B}(\mathcal{X}_B)\times\mathcal{B}(\mathcal{X}_C))$, such that 
	for all $S_{AB}\in \mathcal{B}(\mathcal{X}_A)\times\mathcal{B}(\mathcal{X}_B)$, $P_{ABC}(S_{AB}\times\mathcal{X}_C) = P_{AB}(S_{AB})$, and for all $S_{BC}\in \mathcal{B}(\mathcal{X}_B)\times\mathcal{B}(\mathcal{X}_C)$, $P_{ABC}(\mathcal{X}_A\times S_{BC}) = P_{BC}(S_{BC})$.
	Then we may redefine $(\{f(X_{kt+1}),\dots, f(X_{kt+k}), \delta_{t}\}_{t=0}^\infty, \{\tilde{Y}_t\}_{t=1}^\infty, \{B(t)\}_{t=0}^\infty)$ to be the identity map on the probability space $(\mathcal{X}_{ABC}, \mathcal{B}(\mathcal{X}_{ABC}), P_{ABC})$.
	The gluing lemma ensures that the functional relationship between $\{f(X_{kt+1}),\dots, f(X_{kt+k}), \delta_{t}\}_{t=0}^\infty$ and $\{\tilde{Y}_t\}_{t=1}^{\infty}$ is preserved almost surely, and that~\eqref{eqn:sippre} continues to hold.
	
	Using the same argument, on yet another probability space, one can redefine $$(\{f(X_{kt+1}),\dots, f(X_{kt+k}), \delta_{t}\}_{t=0}^\infty, \{B(t)\}_{t=0}^\infty)$$ together with $\{B(t/k\Xi_{\tau})\}_{t=0}^\infty$ so that
	$\{B(t)\}_{t=0}^\infty$ and $\{B(t/k\Xi_{\tau})\}_{t=0}^\infty$ have the same law as the corresponding elements in a $d$ dimensional standard Brownian motion, while preserving~\eqref{eqn:sippre}.
\end{proof}

\section{Proofs for Section \ref{ssec:tr}}
\subsection{Proof of Theorem \ref{thm:ter}}\label{proof:thmter}

%
 
We begin by considering the asymptotic behavior of $\tN(\varepsilon)$ as $\varepsilon \to 0$.
For simplicity, let $\hat{\mathfrak{M}}_n$ represent $\hat{\mathfrak{M}}_n(f,\pi)$, and $\mathfrak{M}$ represent $\mathfrak{M}(f,\pi)$. With probability 1, $\hat{\mathfrak{M}}_n\to \mathfrak{M}$ as  $n\to\infty$.
Recall that $s(n, \varepsilon) = \varepsilon \hat{\mathfrak{M}}_n(f,\pi) \mathbb{I}(n<n_0) + n^{-1}$.
Let 
$$
n_S(\varepsilon) = \inf\{n\in\mathbb{N}_+:  \, s(n, \varepsilon) \leq  \varepsilon \hat{\mathfrak{M}}_n\} = \inf\{ n \geq n_0: \, n^{-1} \leq \varepsilon \mathfrak{M}_n \}.
$$
Then, with probability 1, as $\varepsilon \to 0$, $n_S(\varepsilon)\to\infty$, and
$
\tN(\varepsilon) \geq n_S(\varepsilon) \to\infty.
$

To obtain the exact rate at which $\tN(\varepsilon)$ tends to infinity, let $L(n) = \{V(S_{\alpha}(n))\}^{1/d} +s(n, \varepsilon)$, so that $\tN(\varepsilon) = \inf\{n\in\mathbb{N}: L(n) \leq  \varepsilon \hat{\mathfrak{M}}_n \}.$ 
Let 
$$a =  \left(\frac{2}{dG(d/2)}\right)^{1/d}(\pi \chi^2_{1-\alpha, d})^{1/2},$$ and recall that
\[
V(S_{\alpha}(n)) = \frac{2\pi^{d/2}}{dG(d/2)}\left(\frac{T^2_{1-\alpha, d, p_n}}{n}\right)^{d/2}|\hat{\Sigma}_n^{\text{BM}}|^{1/2}.
\]
Since $n^{1/2} s(n,\varepsilon) \to 0 $ as $n\to\infty$ and $\Sigma$ is positive definite, we have with probability 1,
\begin{equation}\label{eqn:sr1}
	\begin{aligned}
		n^{1/2} L(n) & = n^{1/2}\left(\frac{2}{dG(d/2)}\right)^{1/d}\left(\frac{\pi T^2_{1-\alpha, d, p_n}}{n}\right)^{1/2}|\hat{\Sigma}_n^{\text{BM}}|^{1/(2d)} + n^{1/2}s(n,\varepsilon) \\
		& \to a |\Sigma|^{1/(2d)}>0, \text{ as } n\to\infty.
	\end{aligned}
\end{equation}
By the definition  of $\tN(\varepsilon)$, 
\begin{equation}\label{eqn:sr2}
	L(\tN(\varepsilon) -1) > \varepsilon \hat{\mathfrak{M}}_{\tN(\varepsilon) -1}, \text{ and } L(\tN(\varepsilon) ) \leq \varepsilon \hat{\mathfrak{M}}_{\tN(\varepsilon)}.
\end{equation}
Since with probability 1, $\hat{\mathfrak{M}}_n\to \mathfrak{M}$, and $N_S(\varepsilon) \to \infty$ as $n\to\infty$, we have with probability 1,
\begin{equation}\label{eqn:sr4}
	\hat{\mathfrak{M}}_{\tN(\varepsilon)}\to \mathfrak{M}, \text{ as } \varepsilon\to0.
\end{equation}
By   \eqref{eqn:sr1}, \eqref{eqn:sr2} and \eqref{eqn:sr4},   with probability 1,
$$\limsup_{\varepsilon\to0} \varepsilon\tN(\varepsilon)^{1/2}  \leq \limsup_{\varepsilon\to0} \frac{\tN(\varepsilon)^{1/2}  L(\tN(\varepsilon) -1)}{\hat{\mathfrak{M}}_{\tN(\varepsilon) -1}} = \frac{a |\Sigma|^{1/(2d)}}{\mathfrak{M}}$$
$$\liminf_{\varepsilon\to0} \varepsilon \tN(\varepsilon)^{1/2}  \geq \liminf_{\varepsilon\to0} \frac{\tN(\varepsilon)^{1/2}  L(\tN(\varepsilon))}{\hat{\mathfrak{M}}_{\tN(\varepsilon)}} = \frac{a |\Sigma|^{1/(2d)}}{\mathfrak{M}}.$$
We have  with probability 1,
\begin{equation}\label{eqn:sr5}
	\lim_{\varepsilon\to0}\varepsilon \tN(\varepsilon)^{1/2} = \frac{a |\Sigma|^{1/(2d)}}{\mathfrak{M}}.
\end{equation}

%

By Theorem \ref{thm:sip}, there exists a discrete time standard Brownian motion $\{B(t)\}_{t=0}^{\infty}$ such that the SIP \eqref{eqn:sip0} holds.
Consider the asymptotic behavior of $\varepsilon B(\tN(\varepsilon))$ as $\varepsilon \to 0$.
It will be shown that 
\begin{equation} \label{eqn:BtN}
	\varepsilon \left\| B(\tN(\varepsilon)) - B\left( N_0(\varepsilon) \right) \right\| \xrightarrow{P} 0,
\end{equation}
as $\varepsilon \to 0$, where 
\[
N_0(\varepsilon) = \left\lceil \frac{a^2 |\Sigma|^{1/d}}{\varepsilon^2 \mathfrak{M}^2} \right\rceil.
\]
Let $\upsilon \in (0,1)$ be arbitrary.
Let $\underline{N} = \lceil (1-\upsilon)N_0(\varepsilon) \rceil$, and $\overline{N} =  \lfloor (1+\upsilon)N_0(\varepsilon) \rfloor$.
By \eqref{eqn:sr5}, with probability~1, for $\varepsilon$ small enough, $\underline{N} \leq N(\varepsilon) \leq \overline{N}$, i.e., 
\[
\lim_{j \to \infty} P \left( \bigcap_{0<\varepsilon < 1/j} [ \underline{N} \leq N(\varepsilon) \leq \overline{N} ] \right) = 1.
\]
It follows that there exists $\varepsilon_{\upsilon} > 0$ such that, whenever $\varepsilon < \varepsilon_{\upsilon}$,
\[
P( \underline{N} \leq N(\varepsilon) \leq \overline{N} ) \geq 1 - \upsilon.
\]
Then, for $\epsilon < \varepsilon_{\upsilon}$, $b > 0$, and $u \in \mathbb{R}^d$,
\[
\begin{aligned}
	&P \left( \varepsilon \left| u^{\top} B(N(\varepsilon)) - u^{\top} B(N_0(\varepsilon)) \right| \geq b \right) \\
	&\leq \upsilon + P \left( \max_{t \in [\underline{N}, N_0(\varepsilon)]} \left| \sum_{i = t}^{N_0(\varepsilon)} u^{\top} C(i) \right| \geq b \varepsilon^{-1} \right) +  P \left( \max_{t \in [N_0(\varepsilon), \overline{N}]} \left| \sum_{i = N_0(\varepsilon)}^t u^{\top} C(i) \right| \geq b \varepsilon^{-1} \right),
\end{aligned}
\]
where $C(i) = B(i+1) - B(i)$ are i.i.d. standard normal random variables.
By Kolmogorov's inequality \citep[see, e.g.,][Theorem D.6.3]{meyn2005markov}, the right hand side is upper bounded by
\[
\upsilon + \frac{2 \varepsilon^2 \|u\|^2}{b^2} \left\{ \upsilon N_0(\varepsilon) + 1 \right\} \leq \upsilon + \upsilon \frac{2 a^2 |\Sigma|^{1/d} \|u\|^2}{\mathfrak{M}^2 b^2} + \frac{4 \varepsilon^2 \|u\|^2}{b^2}.
\]
Since $\upsilon$ is arbitrary, as $\varepsilon \to 0$, \eqref{eqn:BtN} holds.
Evidently, as $\varepsilon \to 0$,
\[
\varepsilon B(N_0(\varepsilon)) \sim \mathcal{N}\left(0, \varepsilon^2 \left\lceil \frac{a^2|\Sigma|^{1/d}}{\varepsilon^2 \mathfrak{M}^2} \right\rceil I_{d \times d} \right) \Rightarrow \mathcal{N}\left(0, \frac{a^2 |\Sigma|^{1/d}}{\mathfrak{M}^2} I_{d \times d} \right),
\]
where $\Rightarrow$ denotes the weak convergence of measures.
Then, by \eqref{eqn:BtN} and Slutsky's theorem,
\begin{equation} \label{eqn:BNep}
\varepsilon B(N(\varepsilon)) \xrightarrow{d} \mathcal{N}\left(0, \frac{a^2 |\Sigma|^{1/d}}{\mathfrak{M}^2} I_{d \times d} \right).
\end{equation}

We shall now consider the asymptotic behavior of $\hat{\theta}_{N(\varepsilon)}$.
By Theorem \ref{thm:sip}, with probability~1, for $\varepsilon$ small enough,
\[
\left\| \sum_{t=1}^{N(\varepsilon)} \varepsilon (f(X_t) - \theta) - \Gamma \varepsilon B(N(\varepsilon)) \right\| \leq M \varepsilon \phi(N(\varepsilon)),
\]
where $\Gamma$ satisfies $\Gamma \Gamma^{\top} = \Sigma$, and $\phi(n)/\sqrt{n} \to 0$ as $n \to \infty$.
By \eqref{eqn:sr5}, with probability~1, for an arbitrary number $b > 0$,
\[
\limsup_{\varepsilon \to 0} \varepsilon \phi(N(\varepsilon)) \leq \limsup_{\varepsilon \to 0}  b \varepsilon N(\varepsilon)^{1/2} = \frac{ab |\Sigma|^{1/(2d)}}{\mathfrak{M}}.
\]
This implies that $\varepsilon \phi(N(\varepsilon)) \to 0$ with probability~1.
In particular, as $\varepsilon \to 0$,
\[
\sum_{t=1}^{N(\varepsilon)} \varepsilon (f(X_t) - \theta) - \Gamma \varepsilon B(N(\varepsilon)) \xrightarrow{P} 0.
\]
Combining this with \eqref{eqn:BNep} and making use of Slutsky's theorem yields
\[
\varepsilon N(\varepsilon) (\hat{\theta}_{N(\varepsilon)} - \theta) \xrightarrow{d} \mathcal{N}\left(0, \frac{a^2 |\Sigma|^{1/d}}{\mathfrak{M}^2} \Sigma \right).
\]
In light of \eqref{eqn:sr5} and Theorem~\ref{thm:strongbm}, another application of Slutsky's theorem gives
\begin{equation} \label{eqn:thetaNnormal}
	N(\varepsilon)^{1/2} (\hat{\Sigma}_{N(\varepsilon)}^{\text{BM}})^{-1/2} (\hat{\theta}_{N(\varepsilon)} - \theta) \xrightarrow{d} \mathcal{N}(0, I_{d \times d}).
\end{equation}

Fix $\alpha \in (0,1)$, and let $\alpha_1$ and $\alpha_2$ be such that $0 < \alpha_1 < \alpha < \alpha_2 < 1$.
Recall that  $T^2_{1-\alpha, d, p_n}\to \chi^2_{1-\alpha, d}$ as $n \to \infty$, so, almost surely, for $\varepsilon$ small enough, $\chi^2_{1-\alpha_2, d} \leq T^2_{1-\alpha, d, p_{N(\varepsilon)}} \leq \chi^2_{1-\alpha_1, d}$.
By \eqref{eqn:thetaNnormal}, with probability~1, for $i = 1,2$,
\[
\lim_{\varepsilon \to 0} P\left( N(\varepsilon) (\hat{\theta}_{N(\varepsilon)} - \theta)^\top (\hat{\Sigma}^{\text{BM}}_{N(\varepsilon)})^{-1} (\hat{\theta}_{N(\varepsilon)} - \theta) < \chi^2_{1-\alpha_i, d} \right) = 1 - \alpha_i,
\]
which implies that
\[
\begin{aligned}
	&\limsup_{\varepsilon \to 0} P\left( N(\varepsilon) (\hat{\theta}_{N(\varepsilon)} - \theta)^\top (\hat{\Sigma}^{\text{BM}}_{N(\varepsilon)})^{-1} (\hat{\theta}_{N(\varepsilon)} - \theta) < T^2_{1-\alpha, d, p_{N(\varepsilon)}} \right) \leq 1 - \alpha_1, \\
	&\liminf_{\varepsilon \to 0} P\left( N(\varepsilon) (\hat{\theta}_{N(\varepsilon)} - \theta)^\top (\hat{\Sigma}^{\text{BM}}_{N(\varepsilon)})^{-1} (\hat{\theta}_{N(\varepsilon)} - \theta) < T^2_{1-\alpha, d, p_{N(\varepsilon)}} \right) \geq 1 - \alpha_2. 
\end{aligned}
\]
But to say $N(\varepsilon) (\hat{\theta}_{N(\varepsilon)} - \theta)^\top (\hat{\Sigma}^{\text{BM}}_{N(\varepsilon)})^{-1} (\hat{\theta}_{N(\varepsilon)} - \theta) < T^2_{1-\alpha, d, p_{N(\varepsilon)}}$ is to say that $\theta \in S_{\alpha}(N(\varepsilon))$.
Since $\alpha_1$ and $\alpha_2$ can be arbitrarily close to $\alpha$, it holds that
\[
\lim_{\varepsilon \to 0} P(\theta \in S_{\alpha}(N(\varepsilon)) ) = 1 - \alpha.
\]
This concludes the proof.

\newpage
\renewcommand{\thepage}{}
\bibliographystyle{plainnat}
\bibliography{ref}

\end{document}